\documentclass[a4paper,10pt,floatfix,superscriptaddress,aps,pra,twocolumn,showpacs,notitlepage,longbibliography,nofootinbib]{revtex4-2}
\usepackage[dvipsnames]{xcolor}
\usepackage{my-colors}
\usepackage[T1]{fontenc}
\usepackage{amsmath}
\usepackage{amsfonts}
\usepackage{amssymb}
\usepackage{amsthm}
\usepackage{enumitem}

\usepackage{pgfplots}
\pgfplotsset{compat=1.18}
\usepgfplotslibrary{fillbetween}
\usetikzlibrary{matrix}

\usepackage{graphicx}
\usepackage{lmodern}
\usepackage{sidecap} 
\usepackage{cancel}
\usepackage{physics}
\usepackage{mathtools}
\usepackage{dsfont}
\usepackage{csquotes}
\usepackage{comment}
\usepackage{ragged2e}
\usepackage{booktabs}
\usepackage[colorlinks]{hyperref}
\hypersetup{
    colorlinks  = true,
    citecolor   = color1!70!black,
    linkcolor   = color2,
    urlcolor    = color3
}
\usepackage{thmtools}
\declaretheorem[parent=section]{theorem}
\usepackage[caption=false]{subfig}

\def\tr{\operatorname{tr}}
\def\mcd{\mathcal{D}}

\DeclarePairedDelimiter{\of}{(}{)}
\DeclarePairedDelimiter{\sof}{\{}{\}}
\DeclarePairedDelimiterX{\midmid}[2]{}{}{#1\;\delimsize\|\;#2}

\newcommand{\phantomcolon}{\mathrel{\phantom{:\mkern-1.2mu}}}
\newcommand{\pvec}[1]{\vec{#1}\mkern2mu\vphantom{#1}}

\newcommand{\ii}{\text{i}}

\begin{document}

\title{A Collection of Pinsker-type Inequalities for Quantum Divergences}
\date{\today}
\author{Kläre Wienecke}
\affiliation{Institut für Theoretische Physik, Leibniz Universität Hannover, Germany}
\author{Gereon Koßmann}
\affiliation{Institute for Quantum Information, RWTH Aachen University, Aachen, Germany}
\author{René Schwonnek}
\affiliation{Institut für Theoretische Physik, Leibniz Universität Hannover, Germany}

\begin{abstract}
 Pinsker's inequality sets a lower bound on the Umegaki divergence of two quantum states in terms of their trace distance. In this work, we formulate corresponding estimates for a variety of quantum and classical divergences including $f$-divergences like Hellinger and $\chi^2$-divergences as well as Rényi divergences and special cases thereof like the Umegaki divergence, collision divergence, max divergence. We further provide a strategy on how to adapt these bounds to smoothed divergences.
\end{abstract} \maketitle

\section{Introduction}

Putting a concrete number on the distance of two quantum states or probability distributions (classical states) can be a quite ambiguous task, as an appropriate answer will strongly depend on the operational tasks at hand.

It is hence no surprise that the field of (quantum and classical) information theory yields a wide range of operationally meaningful and/or mathematically practical measures for quantifying the proximity of two states. Examples considered in this work include families of Rényi divergences, $f$-divergences and smoothed divergences (cf. \, e.g. \cite{Petz1986,Petz2010,Wilde_2014,MllerLennert2013,TomamichelPhDthesis}).

The word divergence is attributed to many of these quantities, as they should quantify how much a state $\sigma$ diverges from a reference state $\rho$. Even though, this term does not have a unique mathematical definition in quantum information processing, there is arguably a consensus that a proper divergence $\mcd(\midmid{\rho}{\sigma})$ should satisfy the data processing inequality: under the action of any channel $\mathcal{C}$\footnote{In this work channels are considered to be completely positive and trace preserving maps \cite{Stinespring1955}.}, the divergence of $\mathcal{C}(\rho)$ relative to $\mathcal{C}(\sigma)$ will decrease when compared to the divergence of $\rho$ relative to $\sigma$. This formalizes the information theoretic intuition that by processing data, i.e.\ by applying $\mathcal{C}$, intrinsic information can only decrease. However, divergences are typically not required to fulfil properties that are somewhat intuitive for geometric distance measures like the triangle inequality or a symmetry between $\rho$ and $\sigma$. 

A distance between states that is usually considered to be ‘the natural distance’ is the \emph{trace distance}
\begin{equation}
T(\rho,\sigma) \coloneqq\frac 12 \Vert \rho -\sigma\Vert_1 =\frac 1 2\tr \left( |\rho-\sigma|\right),
\end{equation} 
often also referred to as total variational distance, normalized 1-norm distance, or statistical distance. It actually ticks boxes in both the geometric and the information-theoretic perspective on distances. Mathematically its a metric. Operationally it has a clear meaning: It quantifies the optimal one-shot success probability for distinguishing $\rho$ from $\sigma$, expressed variationally as an optimization over all quantum effects (cf. \ e.g.\ \cite{Watrous_2018}).

An immediate question that arises when handling divergences, for example in a computation, is: How does the value of a specific divergence $\mcd(\midmid{\rho}{\sigma})$ relate to the value of the trace distance $T(\rho,\sigma)$?
\begin{figure}[t]
    \centering
    \begin{tikzpicture}
    \begin{axis}[
            restrict y to domain=0:12,
            xmin=0,
            xmax=1.1,
            ymin=0,
            ymax=8.5,
            axis lines=left,
            axis line style=thick,
            compat=newest,
            xlabel={$T$}, xlabel style={at={(1,0)}, anchor=west},
            ylabel=$D$, ylabel style={at={(0,1)}, anchor=south, rotate=-90},
            xtick={0,0.5,1},
            set layers,
            axis on top
            ]
        \draw[dashed] (1,0) -- (1,20);
        \addplot[
            only marks,
            color1,
            mark=*,
            fill opacity=0.5,
            draw opacity=0,
            mark size=1.7,
            on layer=axis background
            ] table {relative_5000.csv};
        \addplot[domain=5:500, samples=500, thick, color2, on layer=axis foreground, name path=right bound](
          {((2^x - 1 - x * ln(2)) * (1 - 2^x + 2^x * ln(2^x)))) / ((2^x - 1)^2 * x * ln(2))},
          {(1 - 2^x + 2^x * ln(2^x)) / (2^x - 1)^2 * ln( (x) / (2^x - 1)  )/ln(2) - 2^x * (1 - 2^x + ln(2^x)) / (2^x - 1)^2 * ln((2^x * x)/(2^x - 1)) / ln(2) + ln(ln(2)) / ln(2)}
        );
        \addplot[domain=0:5, samples=350, thick, color2, on layer=axis foreground, name path=left bound](
          {((2^x - 1 - x * ln(2)) * (1 - 2^x + 2^x * ln(2^x)))) / ((2^x - 1)^2 * x * ln(2))},
          {(1 - 2^x + 2^x * ln(2^x)) / (2^x - 1)^2 * ln( (x) / (2^x - 1)  )/ln(2) - 2^x * (1 - 2^x + ln(2^x)) / (2^x - 1)^2 * ln((2^x * x)/(2^x - 1)) / ln(2) + ln(ln(2)) / ln(2)}
        );
        \addplot[domain=0:1, black, thick, on layer=axis foreground]{2/(ln(2))*x^2} node[name=Pinsker]{};
        \draw[color2] (axis cs:1,8) node[name=B_D]{};
        \draw[] (axis cs:0.6605239233325124,4.267934051469621) node[name=p]{};
        
    \end{axis}
    \draw[color2] (B_D) node[anchor=west, yshift=3mm] {$B_D$};
    \draw[black] (Pinsker) node[anchor=west] {Pinsker};
\end{tikzpicture}
    \caption{\label{fig: 1} Improvement of Pinsker's inequality bound comparing the Umegaki divergence $D$ with the trace distance $T$. The dots depict a random sample from the set ${\Omega_{T,D} = \sof{(T(\rho, \sigma), D(\midmid{\rho}{\sigma}))}}$ for a random selection of states $\rho$ and $\sigma$. The red line indicates the optimal convex lower bound $B_D$ on this set while the black line indicates Pinsker's bound. }
\end{figure}
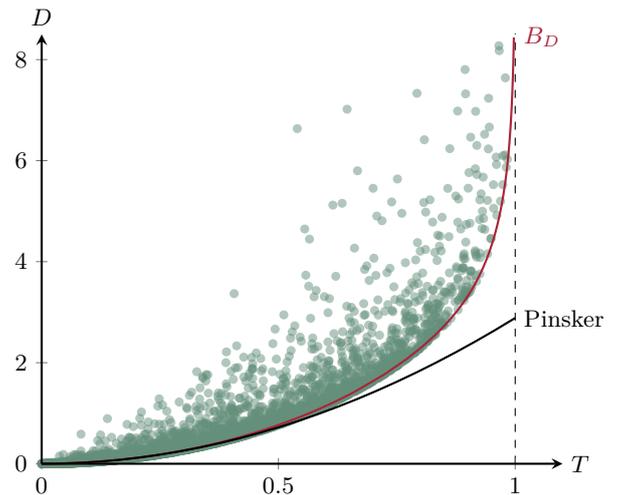
As, for instance, the quantum relative entropy quantifies the asymptotic rate of distinguishing $\rho$ and $\sigma$ in a hypothesis-testing experiment \cite{Hiai1991}, while the trace distance quantifies the corresponding one-shot success probability, it is natural to expect an inequality as a relation between both quantities in best cases. Naturally, there are two different points of view to take here: One can look for upper or lower bounds relating $\mcd$ and $T$. Upper bounds are also known as continuity bounds \cite{Bluhm_2023} and we will refer to lower bounds as Pinsker-type bounds, named after Pinsker's inequality, which in the classical setting provides an answer to the above question by establishing a lower bound on the Kullback-Leibler divergence in terms of the variational distance \cite{Cover_2006}. Analogously, the quantum version of Pinsker's inequality relates the \emph{Umegaki divergence}
\begin{equation}
    D(\midmid{\rho}{\sigma} ) \coloneqq \tr(\rho \log \rho - \rho \log \sigma)
\end{equation}
to the trace distance $T(\rho, \sigma)$:
\begin{equation}
    D( \midmid{\rho}{\sigma}) \geq \frac{1}{2 \ln(2)} \norm{\rho-\sigma}^2_1 = \frac{2}{\ln(2)} T(\rho, \sigma)^2.
\end{equation}

While Pinsker’s bound has become a fundamental tool in quantum information processing, it is only a special case of a broader principle. A complete treatment of how statistical distances relate to each other is beyond Pinsker’s inequality and applies to a much wider class of divergences. Furthermore, while Pinsker's bound is a good choice for small $T$, it is far from optimal for $T$ closer to 1 since the bound diverges at 1, see Figure \ref{fig: 1}. 

Deriving optimal Pinsker-type bounds for a broad class of divergence measures which satisfy the data processing inequality can be done in a surprisingly systematic manner. We develop a general method to obtain such bounds and apply it to various quantum divergences listed in Table \ref{tab: divergences}, the results of which are listed in Table \ref{tab: results}. Furthermore, we show how to adapt the bounds to smoothed divergences including an example also listed in Table \ref{tab: results}.

First, we introduce the strategy of how to find linear and convex (Pinsker-type) bounds using two-dimensional classical states for any divergence fulfilling the data processing inequality and we show how to adapt both linear and convex bounds for smoothed divergences in Section \ref{sec: linear and convex bounds}. Using this strategy, we then present a collection of bounds for a selection of divergences including a detailed analysis of the bounds for Rényi divergences as well as the smoothed max divergence in Section \ref{sec: A collection of bounds} before a brief summary in Section \ref{sec: summary and outlook}.

\section{Linear and Convex bounds}\label{sec: linear and convex bounds}
In order to be in the position to put concrete results on a Pinsker-type inequality, we first have to clarify some basic definitions and geometric/algebraic relations on how such a bound can be formulated. 

The term divergence does unfortunately not have a unique mathematical definition shared amongst all authors. An elegant axiomatic approach by \citeauthor{TomamichelDupuis2013} can be found in \cite{TomamichelDupuis2013}. In fact, for our purposes we actually only need divergences to fulfil the data processing inequality, i.e.\ any divergence in this work ${\mcd:\mathcal{S}(\mathcal{H})_{\neq 0}\times \mathcal{S}(\mathcal{H}) \to \mathbb{R}_{\geq 0}}$, $(\rho, \sigma) \mapsto \mcd(\midmid{\rho}{\sigma})$ needs to satisfy
\begin{equation}\label{eq: data processing}
        \mcd(\midmid{\rho}{\sigma}) \geq \mcd(\midmid{\mathcal{C}(\rho)}{\mathcal{C}(\sigma)}),
\end{equation}
for any channel $\mathcal{C}$, i.e.\ completely positive trace preserving linear map $\mathcal{C}$. Here, we denote by $\mathcal{S}(\mathcal{H})$ the set of states, i.e.\ the set of density operators on a Hilbert space $\mathcal{H}$.  Note that the data processing inequality follows for any divergence satisfying the axioms of \cite{TomamichelDupuis2013}.
Any divergence $\mcd$ considered in this work can also be assumed to be well defined for states coming from any Hilbert space $\mathcal H_n$ of (finite) dimension $n$. 

Applications often also use smoothed divergences of which there exist several definitions, too. In this paper, we define the \emph{$\varepsilon$-smoothed version of $\mcd$} of $\rho$ relative to $\sigma$ as
\begin{equation}
    D^\varepsilon \of*{\midmid{\rho}{\sigma}} \coloneqq \inf_{\rho' \in B^\varepsilon(\rho)} D\of*{\midmid{\rho'}{\sigma}},
\end{equation}
where $\varepsilon > 0$ and the smoothing ball is
\begin{equation}
    B^\varepsilon(\rho)\coloneqq \sof{\rho' \in \mathcal{S}(\mathcal{H}) \mid T(\rho, \rho') \leq \varepsilon}. \footnotemark
\end{equation}
\addtocounter{footnote}{-1}\footnotetext{We consider trace distance smoothed divergences here. Results can be adapted for purified distance smoothed divergences using $B_P^{\sqrt{\varepsilon(2-\varepsilon)}}(\rho) \leq B^\varepsilon_T(\rho) \leq B^\varepsilon_P(\rho)$, where the subscripts indicate the distance used for smoothing \cite{Regula_2025}.}
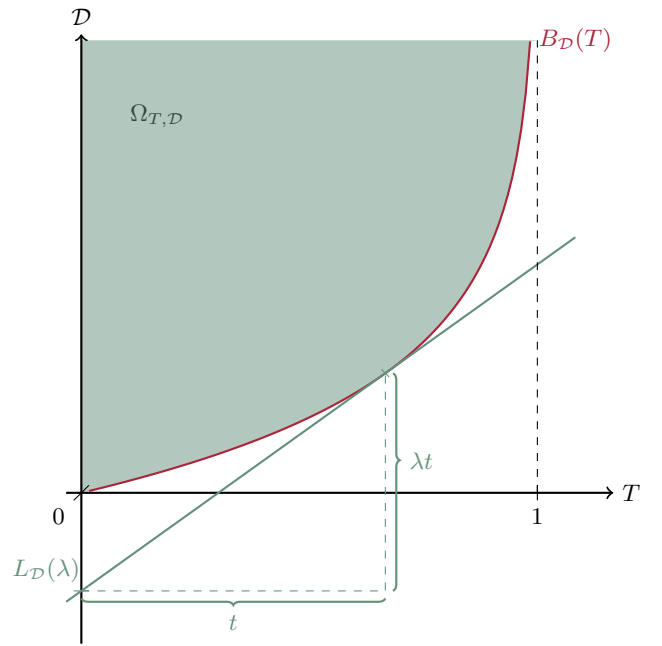
\begin{figure}[t]
    \centering
    \begin{tikzpicture}[]
        \fill [color1!50, domain=0:189/32, samples=100, variable=\x] 
        (0, 6)
        -- plot ({\x}, {ln(6/(6-\x))/ln(2)})
        -- (6, 6)
        -- cycle;
        \draw[thick, ->] (0,-2)--(0,6.08) node[anchor=south]{$\mcd$};
        \draw[thick, ->] (-0.2,0)--(7,0) node[anchor=west]{$T$};
        \draw[](0.1,0.1)--(-0.1,-0.1)node[anchor=north east]{$0$};
        \draw[](6,0.1)--(6,-0.1)node[anchor=north]{$1$};
        \draw[thick, color2] plot [domain=0.1:189/32,samples=100] (\x,{ln(6/(6-\x))/ln(2)});
        \draw[color2](189/32,6)node[anchor=west]{$B_\mcd(T)$};
        \draw[color1] plot[mark=x]coordinates{(4,{ln(3)/ln(2)})} node[anchor=south east]{};
        \draw[thick,color1] plot[smooth,domain=-0.2:6.5] (\x, {1/(2*ln(2))*\x-(2-ln(3))/(ln(2))});
        \draw[color1](-0.1,{-(2-ln(3))/(ln(2))})--(0.1,{-(2-ln(3))/(ln(2))})node[anchor=south east]{$L_\mcd(\lambda)$};
        \draw[color1, dashed](0,{-(2-ln(3))/(ln(2))})--(4,{-(2-ln(3))/(ln(2))})--(4,{ln(3)/ln(2)});
        \draw[dashed](6,0)--(6,6);
        \draw [thick,decoration={brace,mirror},decorate, color1] (0,{-(2-ln(3))/(ln(2))-0.1})--(4,{-(2-ln(3))/(ln(2))-0.1}) node [pos=0.5, anchor=north, yshift=-0.1cm] {$t$};
        \draw [thick, decoration={brace,mirror, aspect=0.6},decorate, color1] (4.1,{-(2-ln(3))/(ln(2))})--(4.1,{ln(3)/ln(2)}) node [pos=0.6, anchor=west, xshift=0.1cm] {$\lambda t$};
        \draw[color1!50!black](1,5) node[]{$\Omega_{T, \mcd}$};
    \end{tikzpicture}
    \caption{\label{fig: tikz legendre} The set $\Omega_{T, \mcd}$ and its linear and convex lower bounds $L_\mcd$ and $B_\mcd$. For each trace distance $t \in (0,1)$, the tangent at $t$ with corresponding slope $\lambda$ intercepts the ordinate at $L_\mcd(\lambda)$. The Legendre transform of the linear lower bound $L_\mcd$ gives the convex lower bound $B_\mcd$.}
\end{figure}

To compare a divergence $\mcd$ to the trace distance we define the two-dimensional set
\begin{equation}
    \Omega_{T, \mcd} \coloneqq \left\{ 
    \left(T(\rho,\sigma), \mcd(\midmid{\rho}{\sigma})\right)  \mid \rho, \sigma \in \mathcal{S}(\mathcal{H}_n), \; n\in \mathbb N
    \right\}
\end{equation}
containing all information on the interrelation between $\mcd$ and $T$.
Pinsker-type inequalities are encoded in the shape of the lower-right boundary curve of this set, see the red line in Figure \ref{fig: 1}.

\begin{table*}[t]
    \centering
    \caption{\label{tab: divergences} Overview of divergences: definitions and binary versions. The $\chi^2$ are written in the version for diagonal states for the sake of simplicity. Note that we take $\log$ to be the logarithm base 2 and $\ln$ to be the natural logarithm.}
        \begin{tabular}{c}
        \toprule
        \begingroup\setlength{\jot}{2ex}
        $
        \begin{alignedat}[t]{7}
        &\text{Divergence Name}   && \mcd(\midmid{\rho}{\sigma})  &&&&  \mcd_\text{bin}(\midmid{r}{s})\\
        \midrule
        &\text{Hellinger-$\alpha$ Divergence  \cite{Hirche_2024}} && D_{H_\alpha} (\midmid{\rho}{\sigma}) &&=\tfrac{1}{\alpha - 1}\of*{\tr\of*{\rho^\alpha \sigma^{1-\alpha}} -1} && \tfrac{1}{\alpha-1}\of*{r^\alpha s^{1-\alpha} + (1-r)^\alpha (1-s)^{1-\alpha}-1} \\
        &\text{Neyman-$\chi^2$ \cite{Hirche_2024, Beigi_2025}} &&\chi_\text{N}^2 \of*{\midmid{\rho}{\sigma}} &&= \tr\bigl(\tfrac{(\rho - \sigma)^2}{\sigma}\bigr) && \tfrac{(r-s)^2}{s(1-s)}\\
        &\text{Pearson-$\chi^2$ \cite{Hirche_2024, Beigi_2025}}&&\chi_\text{P}^2 \of*{\midmid{\rho}{\sigma}} &&= \tr \bigl(\tfrac{(\rho-\sigma)^2}{\rho}\bigr) && \tfrac{(r-s)^2}{r(1-r)}\\
        &\text{Rényi-$\alpha$ Divergence \cite{Tomamichel_2016}} &&D_\alpha(\midmid{\rho}{\sigma}) &&= \tfrac{1}{\alpha-1} \log(\tr(\rho^\alpha\sigma^{1-\alpha})) \qquad&& \tfrac{1}{\alpha-1}\log(r^\alpha s^{1-\alpha} + (1-r)^\alpha (1-s)^{1-\alpha})\\
        &\text{Fidelity Divergence \cite{Tomamichel_2016}} && D_\frac{1}{2}(\midmid{\rho}{\sigma}) &&= -2 \log\bigl(\tr(\sqrt{\rho} \sqrt{\sigma})\bigr) && -2 \log\bigl(\sqrt{rs} + \sqrt{(1-r)(1-s)}\bigr) \\
        &\text{Umegaki Divergence \cite{Tomamichel_2016}}&&D \of*{\midmid{\rho}{\sigma}} &&= \tr (\rho \log(\rho) - \rho \log(\sigma))&& \log\bigl(\of*{\tfrac{r}{s}}^r \bigl(\tfrac{1-r}{1-s}\bigr)^{1-r}\bigr) \\
        &\text{Collision Divergence \cite{Tomamichel_2016}}&&D_2 \of*{\midmid{\rho}{\sigma}} &&= \log\bigl(\tr\bigl(\tfrac{\rho^2}{\sigma}\bigr)\bigr) && \log \bigl( \tfrac{r^2}{s} + \tfrac{(1-r)^2}{1-s}\bigr)\\
        &\text{Max Divergence \cite{Tomamichel_2016}}&&D_\infty\of*{\midmid{\rho}{\sigma}} &&= \inf \sof{\lambda \geq 0 \mid \rho \leq 2^\lambda \sigma} &&\begin{cases}
            \log\Bigl(\tfrac{1-r}{1-s}\Bigr)& r \leq s \\
            \log\of*{\tfrac{r}{s}}&s < r\\
        \end{cases}\\
        &\text{$\varepsilon$-smoothed Max Divergence \cite{Regula_2025, Tomamichel_2016}\quad} &&D^\varepsilon_\infty(\midmid{\rho}{\sigma}) &&= \inf_{\rho' \in B^\varepsilon(\rho)} D_\infty \of*{\midmid{\rho'}{\sigma}} &&\begin{cases}
            \log\bigl(\tfrac{1-r- \varepsilon}{1-s}\bigr) & r + \varepsilon \leq s\\
            \log\of*{\tfrac{r- \varepsilon}{s}}& s < r - \varepsilon\\
            0 & \text{else}
        \end{cases}\\
        \end{alignedat}
        $\endgroup\\
        \bottomrule
        \end{tabular}
\end{table*}

For convex sets $\Omega_{T,\mcd}$ a point $(t,d)$ on the lower-right boundary curve corresponds to the statements\footnote{These types of statements are true for any Pareto convex set. Since convexity implies Pareto convexity, we can make these statements here to give an intuition for lower convex bounds. \cite{ehrgott_2000}}:
\begin{displayquote}
    For all states $\rho$ and $\sigma$ with divergence $\mcd(\midmid{\rho}{\sigma})=d$, the trace distance $T(\rho, \sigma)$ is smaller than $t$.
    
    Conversely, for all states $\rho$ and $\sigma$ with trace distance $T(\rho, \sigma) = t$, the divergence $\mcd(\midmid{\rho}{\sigma})$ is larger than $d$.
\end{displayquote}
It is in general hard to show that these sets $\Omega_{T,\mcd}$ are convex for a general definition of a divergence $\mcd$. However, from experience we know that these sets tend to be convex and that for non-convex sets, it is practical to use the smallest convex set containing $\Omega_{T,\mcd}$ -- its convex hull. In the classical setting, convexity of these sets are for example studied in \cite{Polyanskiy_Wu_2025}. How do we find a convex lower bound on these sets? Let us now describe a strategy on how to obtain convex lower bounds by stopping over at linear lower bounds. 

Every convex set can be seen as the intersection of closed half-spaces, which in turn can be defined via linear inequalities. The linear inequalities in question here are
\begin{equation}
    \mcd(\midmid{\rho}{\sigma}) - \lambda T(\rho, \sigma) \geq L_\mcd(\lambda),
    \label{eq: linear inequality}
\end{equation}
where for every trace distance $T$, we find a corresponding slope $\lambda$ and a lower bound $L_\mcd$, see Figure \ref{fig: tikz legendre}. Naturally, $L_\mcd$ is found by solving the optimization problem 
\begin{align}\label{eq: definition L_D}
    &L_{\mcd}(\lambda)=\inf_{\rho, \sigma} \sof{\mcd(\midmid{\rho}{\sigma}) - \lambda T(\rho, \sigma )}\\ 
    &\operatorname{s.th.} \rho, \sigma \in \mathcal{S}(\mathcal{H}_n), \, n\in \mathbb{N} .\nonumber
\end{align}
Equation \eqref{eq: linear inequality} holds in particular for the largest $\lambda$ and we thus define
\begin{equation}\label{eq: definition B_D}
    B_\mcd(T ) \coloneqq \sup_\lambda \sof{L_\mcd(\lambda) + \lambda T}
\end{equation}
as the convex lower bound, which is a Legendre-type transform of $L_\mcd$. With that, we have constructed a convex lower bound $B_\mcd$ to any divergence $\mcd$ as a function of the trace distance:
\begin{equation}
    \mcd(\midmid{\rho}{\sigma}) \geq B_\mcd \left(T (\rho, \sigma) \right).
\end{equation}
Both the linear bound $L_\mcd$ and the convex lower bound $B_\mcd$ are illustrated in Figure \ref{fig: tikz legendre}. 

It is clear that the problem of finding a convex lower bound $B_\mcd$ has become a two-step problem: First, one finds the linear lower bound $L_\mcd$ by solving the minimization problem \eqref{eq: definition L_D}, from which then, in the second step, $B_\mcd$ is obtained using \eqref{eq: definition B_D}. Note that, in some applications it may even be more useful to use the linear bound and not go all the way up to the convex bound.

Optimizing over all states from any dimension, possibly of very large size, in order to find the linear lower bound  \eqref{eq: definition L_D} in the first step may seem like quite a difficult task. Fortunately, dealing with high dimensions can be avoided for the computation of optimal bounds. The central observation, used throughout this work, is that the optimum in \eqref{eq: definition L_D} is attained on binary jointly diagonal states. We define the \emph{binary version of a divergence} $\mcd$ as 
\begin{equation}\label{eq:d2}
    \mcd_\text{bin}(\midmid{r}{s} ) \coloneqq \mcd \of*{
    {
    \begin{pmatrix}
        r & 0\\ 0 & 1-r
    \end{pmatrix}} \,\bigg\|\,
    {  
    \begin{pmatrix}
        s & 0\\ 0 & 1-s
    \end{pmatrix}
    }}
\end{equation}
where $r, s \in [0,1]$. 
It is well known \cite{Wilde_2016} that the value of the trace distance in attained by the Helstrom measurement. The corresponding measurement channel can therefore be used to reduce the value of a divergence while keeping the value of a trace distance constant. This simple observation is exactly the argument needed to characterize the optimal bound in a Pinsker inequality. Formalizing this intuition leads to the following  theorem.

\begin{theorem}\label{thm: main theorem}
For any divergence $\mcd$ that fulfils the data processing inequality $\mcd(\midmid{\rho}{\sigma}) \geq \mcd(\midmid{\mathcal{C}(\rho)}{\mathcal{C}(\sigma)}) \;\forall \mathcal{C}$, the optimal convex lower bound $B_\mcd$ is attained by two-dimensional classical states. The optimal linear lower bound is given by the optimization
    \begin{align}
       &L_{\mcd}(\lambda)=\inf_{r,s} \sof{\mcd_\textup{bin}(\midmid{r}{s})- \lambda (r-s) }\quad \\ &\text{s.\ th.\ } 0\leq s \leq r\leq 1.\nonumber
    \end{align}
\end{theorem}
\begin{proof}
    see Appendix \ref{Appendix: Proof of main theorem}.
\end{proof}

This strategy for finding linear and convex lower bounds is not directly applicable for smoothed divergences, because it is a priori not clear how the smoothing effects the minimizations. However, both linear and convex lower bounds can be adapted for smoothed divergences as stated in the following theorem.

\begin{theorem}\label{thm: smoothed bounds}
    Let $\mcd: \mathcal{S}(\mathcal{H})_{\neq 0} \times \mathcal{S}(\mathcal{H}) \to \mathbb{R}$ be a divergence that fulfils 
    \begin{itemize}
        \item[(i)] data-processing 
        \item[(ii)] non-negativity
        \item[(iii)] convexity in the first argument
        \item[(iv)] faithfulness, i.e. $\mcd(\midmid{\rho}{\rho}) = 0$. 
    \end{itemize}
    Let $L_\mcd$ be the optimal linear lower bound and $B_\mcd$ be the optimal convex lower bound for this divergence. Additionally, let $\mcd^\varepsilon$ be the $\varepsilon$-smoothed version of $\mcd$, which also fulfils the data-processing inequality.
    
    Then, the optimal linear lower bound of $\mcd^\varepsilon$ is given by
    \begin{equation}
        L_{\mcd^\varepsilon}(\lambda)= \begin{cases}
            0 & \text{for } 0 \leq \lambda \leq \lambda_\varepsilon, \\
            L_\mcd(\lambda) - \lambda \varepsilon  & \text{for }\lambda_\varepsilon < \lambda \leq \lambda_\text{max},
        \end{cases}
    \end{equation}
    where $\lambda_\varepsilon = \left. \frac{\dd B_\mcd}{\dd T}\right\vert_{T=\varepsilon}$ and $\lambda_\text{max} = \left.\frac{\dd B_\mcd}{\dd T}\right\vert_{T=1-\varepsilon}$. The optimal convex lower bound is given by
    \begin{equation}
        B_{\mcd^\varepsilon}(T)= \begin{cases}
            0 & \text{for } 0 \leq T \leq \varepsilon,\\
            B_\mcd(T - \varepsilon) & \text{for }\varepsilon < T \leq 1.
        \end{cases}
    \end{equation}
\end{theorem}
\begin{proof}
    see Appendix \ref{Appendix: Proof of smoothed theorem}.
\end{proof}
Note that $\varepsilon$-smoothing shifts the linear lower bound linearly in $\varepsilon$ and the convex lower bound by $\varepsilon$. Both bounds vanish in the region where the divergence can attain zero. Figure \ref{fig: convex bound smoothed divergence} compares the convex lower bounds of the smoothed and unsmoothed max divergence.
\begin{figure}
    \centering
    \begin{tikzpicture}
        \begin{axis}[
            restrict y to domain=0:12,
            xmax=1.1,
            ymax=8.5,
            axis lines=left,
            axis line style=thick,
            compat=newest,
            xlabel=$T$, xlabel style={at={(1,0)}, anchor=west},
            ylabel=$D^\varepsilon_\infty$, ylabel style={at={(0,1)}, anchor=south, rotate=-90},
            xtick={0,0.2,0.5,1},
            xticklabels={0,$\varepsilon$,0.5,1},
            axis on top
            ]
            \draw[dashed] (1,0) -- (1,20);
            \addplot[domain=0:1, color2, thick, samples=400, dashed] {ln(1/(1-x))/ln(2)} node[name=B_D]{};
            \addplot[domain=0:1, color2, thick, name path=right bound, samples=400] {ln(1/(1-x+.2))/ln(2)} node[name=B_D_smooth]{};
            \addplot[domain=0:0.2, color2, thick, name path=left bound] {0*x};
            \path[name path=clippath right] (0.19,0)--(0.19,8.5) -- (1,8.5);
            \path[name path=clippath left] (0,0)--(0,8.5) -- (0.2,8.5);
            \addplot[color1] fill between [of=right bound and clippath right, soft clip={domain=0:1}];
            \addplot[color1] fill between [of=left bound and clippath left, soft clip={domain=0:0.2}];
        \end{axis}
    \draw[color2] (B_D) node[anchor=south] {$B_{D_\infty}$};
    \draw[color2] (B_D_smooth) node[anchor=west] {$B_{D^\varepsilon_\infty}$};
    \end{tikzpicture}
    \caption{Convex lower bound for smoothed quantum max divergence. The bound for the smoothed max divergence (solid red line) is obtained by shifting the bound for the unsmoothed max divergence (dashed red line) by $\varepsilon$ to the right and cutting off at $D^\varepsilon_\infty = 0$ and $T=1$.}
    \label{fig: convex bound smoothed divergence}
\end{figure}
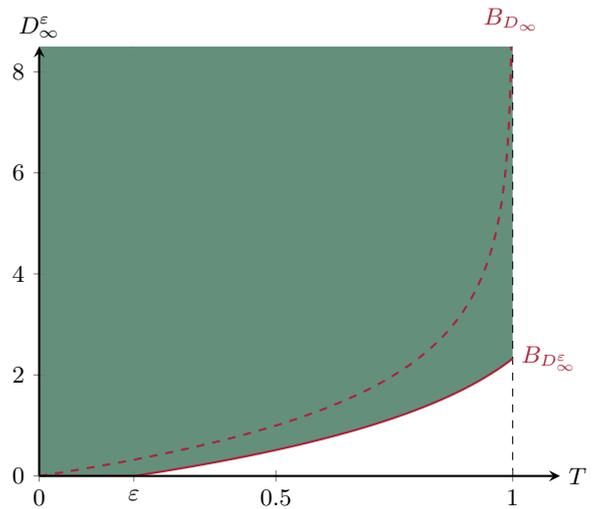

Using these two theorems, we can now obtain bounds for the collection of divergences found in Table \ref{tab: divergences} and improve on Pinsker's inequality.

\section{A collection of bounds}\label{sec: A collection of bounds}
\begin{table*}[t]
    \centering
    \caption{\label{tab: results} Overview of the linear and convex bounds $L_\mcd$ and $B_\mcd$.}
    \begin{tabular}{c}
    \toprule
    \begingroup\setlength{\jot}{2ex}
    $
    \begin{alignedat}[t]{9}
        &\mcd   && L_\mcd(\lambda) && \lambda-\text{range}  &&  B_\mcd(T) && T-\text{range}\\
        \midrule
        &D_{H_\alpha} &&\begin{cases}
                \text{numerically: \ref{sec: numerics}} \\
                \frac{1}{1-\alpha}\of*{1+\lambda(\alpha - 1 - \alpha \lambda^{-\frac{1}{\alpha}})}
            \end{cases} && \begin{array}{l}
                    0\leq \lambda < \of*{\frac{\alpha}{\alpha-1}}^\alpha\\
                    \of*{\frac{\alpha}{\alpha-1}}^\alpha\leq \lambda \leq \infty
        \end{array}\,\,&& \begin{cases}
                \text{numerically: \ref{sec: numerics}} \\
                \frac{1-(1-T)^{1-\alpha}}{1-\alpha} 
            \end{cases} && \begin{array}{l}
        0\leq T < \frac{1}{\alpha}\\
        \frac{1}{\alpha}\leq T \leq 1
        \end{array}\\
        &\chi_\text{N}^2 &&\begin{cases}
                -\frac{\lambda^2}{16}\\
                -(\sqrt{\lambda}-1)^2
            \end{cases} && \begin{array}{l}
                    0 \leq \lambda \leq 4\\
                    4 \leq \lambda \leq \infty
        \end{array} && \begin{cases}
                4T^2  \\
                \frac{T}{1-T}
            \end{cases} && \begin{array}{l}
        0 \leq T \leq \tfrac{1}{2}\\
        \tfrac{1}{2} \leq T \leq 1
        \end{array}\,\,\\
        &\chi_\text{P}^2 &&\begin{cases}
                -\frac{\lambda^2}{16}\\
                -(\sqrt{\lambda}-1)^2
            \end{cases} && \begin{array}{l}
                    0 \leq \lambda \leq 4\\
                    4 \leq \lambda \leq \infty
        \end{array} && \begin{cases}
                4T^2  \\
                \frac{T}{1-T}
            \end{cases} && \begin{array}{l}
        0 \leq T \leq \tfrac{1}{2}\\
        \tfrac{1}{2} \leq T \leq 1
        \end{array}\,\,\\
        &D_{\alpha\in [1,\infty]} &&\begin{cases}
                \text{numerically: \ref{sec: numerics}} \\
                \frac{1}{\ln(2)} - \lambda + \log( \lambda \ln(2))
            \end{cases} && \begin{array}{l}
                    0\leq \lambda < \frac{\alpha}{(\alpha-1)\ln(2)}\\
                    \frac{\alpha}{(\alpha-1)\ln(2)}\leq \lambda \leq \infty
        \end{array}\,\,&& \begin{cases}
                \text{numerically: \ref{sec: numerics}} \\
                \log{\Bigl(\frac{1}{1-T}\Bigr)} 
            \end{cases} && \begin{array}{l}
        0\leq T < \frac{1}{\alpha}\\
        \frac{1}{\alpha}\leq T \leq 1
        \end{array}\\
        &D_{\frac{1}{2}} && \tfrac{1}{\ln(2)} \of*{1-\sqrt{1+\lambda^2\ln(2)^2}+\ln\of*{\tfrac{1+\sqrt{1+\lambda^2\ln(2)^2}}{2}} }&&\,\,0 \leq \lambda \leq \infty && \log \of*{\tfrac{1}{1-T^2}} && \,\,0 \leq T \leq 1 \\
        &D&&\begin{array}{l}\log\of*{\of*{\tfrac{r^*}{s^*}}^{r^*} \of*{\tfrac{1-r^*}{1-s^*}}^{1-r^*}}\\
        \text{ with } (r^*, s^*) = (\tfrac{2^\lambda(2^\lambda - 1 - \lambda \ln(2))}{(2^\lambda-1)^2}, \tfrac{1}{1-2^\lambda} + \tfrac{1}{\lambda \ln(2)})\end{array} && \,\,0 \leq \lambda \leq \infty && \text{parametrized: \ref{sec: parametrized bound}} &&\,\, 0 \leq T \leq 1 \\
        &D_2&&\begin{cases}
                \frac{1}{\ln(4)}\of*{-2-\sqrt{4-\lambda^2\ln(2)^2} - 2 \ln\of*{\frac{2-\sqrt{4-\lambda^2\ln(2)^2}}{4}}} \\
                \frac{1}{\ln(2)} - \lambda + \log( \lambda \ln(2))
            \end{cases} && \begin{array}{l}
                    0\leq \lambda < \frac{2}{\ln(2)}\\
                    \tfrac{2}{\ln(2)}\leq \lambda \leq \infty 
        \end{array} && \begin{cases}
                \log(1+4T^2) \\
                \log\of*{\tfrac{1}{1-T}}
            \end{cases} && \begin{array}{l}
        0 \leq T \leq \tfrac{1}{2}\\
        \tfrac{1}{2} \leq T \leq 1
        \end{array}\\
        &D_\infty \quad&&\tfrac{1}{\ln(2)} - \lambda + \log( \lambda \ln(2)) &&\,\, 0 \leq \lambda \leq \infty &&\log\of*{\tfrac{1}{1-T}} &&\,\, 0 \leq T \leq 1\\
        &D^\varepsilon_\infty \quad&&\begin{cases}
                0 \\
                \tfrac{1}{\ln(2)} - \lambda + \log( \lambda \ln(2)) - \lambda \varepsilon 
            \end{cases} && \begin{array}{l}
                    0 \leq \lambda \leq \frac{1}{\ln(2)(1-\varepsilon)} \\
                    \frac{1}{\ln(2)(1-\varepsilon)} < \lambda \leq \frac{1}{\ln(2)\varepsilon} 
        \end{array} && \begin{cases}
                0 \\
                \log\of*{\tfrac{1}{1-(T-\varepsilon)}}
            \end{cases} && \begin{array}{l}
            0 \leq T \leq \varepsilon \\
        \varepsilon < T \leq 1
        \end{array}\\
        \end{alignedat}
    $\endgroup\\
    \bottomrule
    \end{tabular}
\end{table*}

There are numerous divergences found in the literature for which one could try to find convex Pinsker-type and linear bounds. Our divergence selection can be found in Table \ref{tab: divergences}. It includes the definitions of the divergences as well as their binary versions. Proofs that all these divergences indeed satisfy the data-processing inequality and we thus can use Theorem \ref{thm: main theorem} can be found in \cite{Hiai_2017,Wilde_2018, Audenaert_2015, Anshu_2019}. 

Using the strategy introduced in the previous section, we obtain the bounds presented in Table \ref{tab: results}. Probably the most interesting cases are that of the Rényi divergence and that of the smoothed max divergence. We showcase both in detail in Sections \ref{sec: Example Renyi divergence} and \ref{sec: example smoothed}. Not only does the Rényi case include the Umegaki, collision and max divergence for specific choices of $\alpha$,  but it also highlights the questions that arise when applying our strategy. 

Some noteworthy observations on the bounds in Table \ref{tab: results} are:
\begin{enumerate}[label={(\roman*)}]
    \item Most of the bounds for unsmoothed divergences are piecewise functions. The pieces transition into each other continuous-differentiably. The example in Section \ref{sec: Example Renyi divergence} illustrates where this piecewise nature comes from.
    \item The bounds for both $\chi^2$-divergences are in fact the same: This is obvious from the fact that their binary versions $\mcd_\text{bin}$ are the same under permutation of $r$ and $s$ and hence, we find the same value in the combined $r,s$ optimization.
    \item For all Rényi divergences with $\alpha \in [1,\infty]$ the convex lower bound for $\frac{1}{\alpha} \leq T \leq 1 $ is identical. In particular, it is valid for the whole range of $T$ for the max divergence ($\alpha \to \infty$). For the collision divergence ($\alpha = 2$), it holds for half the $T$-range, i.e.\  $\frac{1}{2} \leq T \leq 1$.
    \item To our knowledge, not all bounds can be calculated analytically as functions of $T$. For these cases, we present parametrized and numerical solutions.
    \item Clearly, the bound for the Umegaki divergence (i.e.\  $D_\alpha$ in the limit $\alpha \to 1$) is not only a bound to the Umegaki divergence, but also to all other Rényi divergences of order $\alpha \in [1,\infty]$, albeit not a tight one for the whole range, see also Figure \ref{fig: numerical convex bounds}.
    \item The bounds for the $\varepsilon$-smoothed max divergence are a shifted and slightly adjusted version of the bounds of the max divergence. This is discussed further in Section \ref{sec: example smoothed}.
    \item The bound for the fidelity divergence we obtained here is known by the name of \emph{Bertagnolle-Huber inequality} in the classical setting \cite{Tsybakov_2009}.
\end{enumerate}

\subsection{Example: Rényi Divergence}\label{sec: Example Renyi divergence}

\begin{figure}
  \centering
  \includegraphics{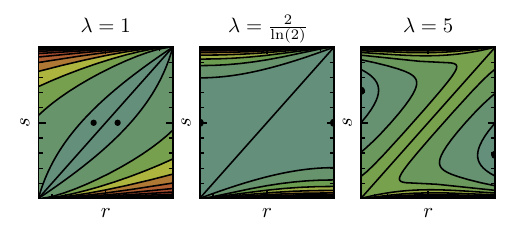}
  \caption{\label{fig: three_mini_plots}$\xi_{\alpha, \lambda}$ on the unit square for $\alpha = 2$ and different values of $\lambda$. The two minima (black dots) move to the vertical boundaries as $\lambda$ increases. For $\lambda \geq \frac{2}{\ln(2)}$ the minima lie at the boundary. Due to the symmetry it is clear that the values of both minima coincide.}
\end{figure}

To illustrate how to use our main theorem and find linear and convex bounds $L_\mcd$ and $B_\mcd$, we go through the example of the Rényi divergence. The Rényi divergence is of special interest in several regards: First, its edge cases are the Umegaki divergence and the max divergence, both of which find use in numerous applications. Second, it is a nice illustration of the cases that occur when calculating linear and convex bounds.

There are different quantum Rényi divergences to be found in the literature, for an overview and discussion see for example \cite{Audenaert_2015}. The two most commonly used ones are the Petz and the sandwiched Rényi divergences of order $\alpha \in[\frac{1}{2},1)\cup(1, \infty)$. The Petz Rényi divergence of order $\alpha$ is given by
\begin{equation}
    D_\alpha \of*{\midmid{\rho}{\sigma}} \coloneqq 
    \tfrac{1}{\alpha - 1}\log(\tr(\rho^\alpha \sigma^{1-\alpha}))
\end{equation}
and the sandwiched Rényi divergence of order $\alpha$ by
\begin{align}
    \widetilde{D}_\alpha \of*{\midmid{\rho}{\sigma}} &\coloneqq 
    \tfrac{1}{\alpha - 1}\log\norm*{\sigma^{\frac{1-\alpha}{2 \alpha}} \rho  \sigma^{\frac{1-\alpha}{2 \alpha}}}_\alpha^\alpha\\
     & \phantomcolon = \tfrac{1}{\alpha - 1}\log\of*{\tr\of*{\of*{ \sigma^{\frac{1-\alpha}{2 \alpha}} \rho  \sigma^{\frac{1-\alpha}{2 \alpha}}}^\alpha}}
\end{align}
for $(\alpha < 1 \text{ and }\rho \not\perp \sigma) \text{ or }\rho \ll \sigma$ and otherwise by $ D_\alpha \of*{\midmid{\rho}{\sigma}} = \widetilde{D}_\alpha \of*{\midmid{\rho}{\sigma}} = \infty$.

Naturally, the binary versions of all Rényi divergences coincide to be
\begin{equation}
    D_{\alpha, \text{bin}} = \frac{1}{\alpha-1}\log(r^\alpha s^{1-\alpha} + (1-r)^\alpha (1-s)^{1-\alpha})
\end{equation}
and hence for obtaining lower bounds it does not matter which one is used. This is of course the case because all quantum Rényi divergences are extensions of the same classical Rényi divergence \cite{Tomamichel_2016}. 

We will see that it is possible to calculate linear and convex lower bounds analytically as functions of the trace distance for the fidelity divergence ($\alpha = \frac{1}{2}$), the collision divergence ($\alpha = 2$), max divergence ($\alpha = \infty$) and Rényi divergences for general $\alpha \in [1,\infty]$ for the restricted region $\frac{1}{\alpha} \leq T \leq 1$. We go through the cases for $\alpha \in [1, \infty]$ here and leave the calculation for $\alpha = \frac{1}{2}$ to Appendix \ref{Appendix: alpha 1/2}.

As a first step we use theorem \ref{thm: main theorem} to find the optimal linear lower bound $L_{D_\alpha}$ by solving the optimization problem
\begin{align}
       &L_{D_\alpha}(\lambda)=\inf_{r,s} \sof{D_{\alpha,\text{bin}}(\midmid{r}{s})- \lambda (r-s) }\quad \\ &\operatorname{s.th. } 0\leq s \leq r\leq 1.\nonumber
\end{align}
For this purpose, we define the function
\begin{equation}
    \xi_{\alpha, \lambda}(r,s) \coloneqq D_{\alpha,\text{bin}}(\midmid{r}{s})- \lambda \abs{r-s},
\end{equation}
of which we need to find the minimum within the range $0\leq s \leq r \leq 1$.

Figure \ref{fig: three_mini_plots} shows $\xi_{\alpha, \lambda}$ for the case $\alpha = 2$ and different values of $\lambda$.

Clearly, the function $\xi_{\alpha, \lambda}$ is symmetric under the interchange $(r,s) \leftrightarrow (1-r,1-s)$. Since we consider the range $0 \leq s \leq r \leq 1$ in the optimization, it suffices to consider the lower right triangles in Figure \ref{fig: three_mini_plots}. Note that in Figure \ref{fig: three_mini_plots}, the minimum wanders to the $(r=1)$-boundary of the triangle and beyond as $\lambda$ grows. However, because we restrict the domain to $r,s \in [0,1]$, the effective local minimum lies at the boundary for sufficiently large $\lambda$. We therefore study two cases: 
\begin{enumerate}
    \item the minimum lies at the $(r=1)$-boundary,
    \item the minimum lies within the lower right triangle rather than at the boundary.
\end{enumerate}

\subsubsection{Minimum at the Boundary}
In this first case, we know that the position of the minimum $(r^*, s^*)$ lies at the $(r=1)$-boundary, i.e.\ we know that $r^* = 1$. The value of $s^*$ is found by solving
\begin{equation}
    \frac{\partial \xi_{\alpha, \lambda} (r=1,s)}{\partial s} \overset{!}{=} 0,
\end{equation}
yielding 
\begin{equation}
    (r^*, s^*) = \of*{1, \tfrac{1}{\lambda \ln(2)}}.
\end{equation}
Consequently, our linear bound is
\begin{equation}
    L_{D_\alpha}(\lambda) = \xi_{\alpha, \lambda} (r^*, s^*) = \frac{1}{\ln(2)} - \lambda + \log( \lambda \ln(2)).
\end{equation}

Naturally, the domain of this bound is not the whole positive real axis, $\lambda \in [0, \infty ]$, but rather a right-side section of it, $\lambda \in [\lambda_\text{crit},\infty]$. The critical intersection point $\lambda_\text{crit}$ is found by determining the point at which the global minimum lies exactly at the boundary, in which case the $r$-derivative vanishes in addition to the $s$-derivative, which is already zero. We find
\begin{align}
    \left.\frac{\partial \xi_{\alpha, \lambda}(r,s)}{\partial r}\right\vert_{s=s^*, r=1} \overset{!}{=} 0 \\
    \Rightarrow \lambda_\text{crit} =  \frac{\alpha}{(\alpha - 1) \ln(2)}.
\end{align}

The convex lower bound $B_{D_\alpha}$ is then found by solving the minimization problem \eqref{eq: definition B_D}. Via
\begin{equation}
    \frac{\partial}{\partial \lambda}(L_{D_\alpha}(\lambda) + \lambda T) \overset{!}{=} 0 
\end{equation}
we find
\begin{equation}
    \lambda^* = \frac{1}{(1-T)\ln(2)},
\end{equation}
an therefore
\begin{align}
    B_{D_\alpha}(T) &= L_{D_\alpha}(\lambda^*) + \lambda^* T \nonumber \\
        &= \log\of*{\frac{1}{1-T}},
\end{align}
for $T \in [\frac{1}{\alpha}, 1]$ and $\alpha \in [1,\infty]$.

\subsubsection{Minimum not at the Boundary}
In this case, there is no simplification to the minimization problem similar to that of the first case. Unfortunately, we do not think that this problem is solvable analytically for general $\alpha \in [1,\infty]$ as it involves finding the roots of an exponential polynomial \cite{bell1934exponential}. There is, however, a special case in which it is solvable, namely for $\alpha = 2$, i.e.\  for the collision divergence.

For the collision divergence, we obtain the linear bound
\begin{equation}
    L_{D_2} (\lambda) = \xi_{2, \lambda} (r^*,s^*) \qquad \text{for } 0\leq \lambda \leq \frac{2}{\ln(2)},
\end{equation}
by calculating
\begin{equation}
    \vec{\nabla} \xi_{2, \lambda} (r,s) \overset{!}{=} 0,
\end{equation}
determining the position of the minimum to be at
\begin{equation}
    (r^*, s^*) = \bigg(\frac{2+\lambda \ln(2) - \sqrt{4 - \lambda^2 \ln(2)^2}}{\lambda \ln(4)},\frac{1}{2}\bigg).
\end{equation}
Thus, the linear lower bound for the range ${\lambda \in[0,\frac{2}{\ln(2)}]}$ is 
\begin{alignat}{1}
    L_{D_2} (\lambda) = \frac{1}{\ln(4)}\bigg(&-2-\sqrt{4-\lambda^2\ln(2)^2} \nonumber\\ 
    &- 2 \ln\of*{\tfrac{2-\sqrt{4-\lambda^2\ln(2)^2}}{4}}\bigg).
\end{alignat}

The convex lower bound then is found similarly to how it was in the first case. We find
\begin{align}
    B_{D_2} (T) &= \sup_{0 \leq \lambda \leq \frac{2}{\ln(2)}} \sof{L_{D_2}(\lambda) + \lambda T} \nonumber \\
    &= \log(1+4 T^2) \qquad \text{for } 0\leq T \leq \frac{1}{2}.
\end{align}

\subsection{Example: Smoothed Quantum Max Divergence}\label{sec: example smoothed}
To demonstrate how to calculate Pinsker-type bounds for smoothed divergences using Theorem \ref{thm: smoothed bounds}, we look at the example of the smoothed max divergence, which meets all requirements for Theorem \ref{thm: smoothed bounds}, see \cite{Tomamichel_2016, Anshu_2019}. We have already calculated the linear and convex lower bounds for the quantum max divergence in Section \ref{sec: Example Renyi divergence}:
\begin{align}
    L_{D_\infty}(\lambda) &= \frac{1}{\ln(2)} - \lambda + \log( \lambda \ln(2))\\
    B_{D_\infty}(T) &= \log\of*{\frac{1}{1-T}}.
\end{align}
Using Theorem \ref{thm: smoothed bounds}, we easily obtain the bounds for the $\varepsilon$-smoothed quantum max divergence as shifted versions of the unsmoothed bounds:
\begin{align}
    L_{D^\varepsilon_\infty}(\lambda) &= \begin{cases}\displaystyle
            0 &  \lambda \in [0, \lambda_\varepsilon], \\
            \tfrac{1}{\ln(2)} - \lambda + \log( \lambda \ln(2)) - \lambda \varepsilon  &\lambda \in (\lambda_\varepsilon, \tilde\lambda],
        \end{cases}\\
    B_{D^\varepsilon_\infty}(T) &= \begin{cases}
            0 & 0 \leq T \leq \varepsilon, \\ \displaystyle
            \log\of*{\tfrac{1}{1-(T-\varepsilon)}} &  \varepsilon < T \leq 1,
        \end{cases}
\end{align}
where $\lambda_\varepsilon = \frac{1}{\ln(2)(1-\varepsilon)}$ and $\tilde\lambda = \tfrac{1}{\ln(2)\varepsilon}$.

Figure \ref{fig: convex bound smoothed divergence} shows the convex lower bound for an exemplary $\varepsilon = 0.2$, comparing it to the bound of the unsmoothed divergence. To find the bound of the smoothed version, the bound of the unsmoothed version is shifted to the right by $\varepsilon$ and set to zero in the range $T \in [0,\varepsilon]$. This is because divergences are positive and within this range, the smoothing ball around $\rho$ includes $\sigma$, causing the minimum divergence value to be zero.

Furthermore, note that unlike convex lower bounds of unsmoothed divergences, the bound for the smoothed divergence does not diverge at $T=1$. This can be easily understood on the Bloch sphere. There, trace distances close to one correspond to states that lie close to opposite poles, in our case on the $z$-axis. Smoothing allows one of the states to vary slightly along the axis within the smoothing ball. Consequently, the divergence can be finite at $T=1$ because it is effectively computed on states that lie closer together in trace distance than $T=1$.

\subsection{Parametrized Bounds}\label{sec: parametrized bound}
For the Umegaki divergence $D$ the linear bound can be found analytically, however, to our knowledge, the convex bound can not. To still give the bound in an analytical form, we can understand it as a parametrized curve in the $TD$-plane instead of as a function of the trace distance. This strategy was pursued in an earlier work by \citeauthor{Fedotov} in \cite{Fedotov}, though using the natural logarithm for the Umegaki divergence. Since we use base-2-logarithm throughout, we give the corresponding parametrized convex lower bound here fore completeness:
\begin{align}
    T(t) &= \frac{(2^t-1-\ln(2^t))(1-2^t+2^t\ln(2^t))}{(2^t-1)^2\ln(2^t)} \\
    D(t) &=\frac{1-2^t+2^t\ln(2^t)}{(2^t-1)^2}\log\of*{\tfrac{t}{2^t-1}}\\
     &\ \ \ \ \  -2^t\frac{1-2^t+\ln(2^t)}{(2^t-1)^2} \log\of*{\tfrac{2^t t}{2^t-1}} + \log(\ln(2))\nonumber
\end{align}
for $t \geq 0$. Both the derivation for the linear as well as for the parametrized convex lower bound can be found in Appendix \ref{Appendix: relative entropy}.

\subsection{Numerical Bounds}\label{sec: numerics}
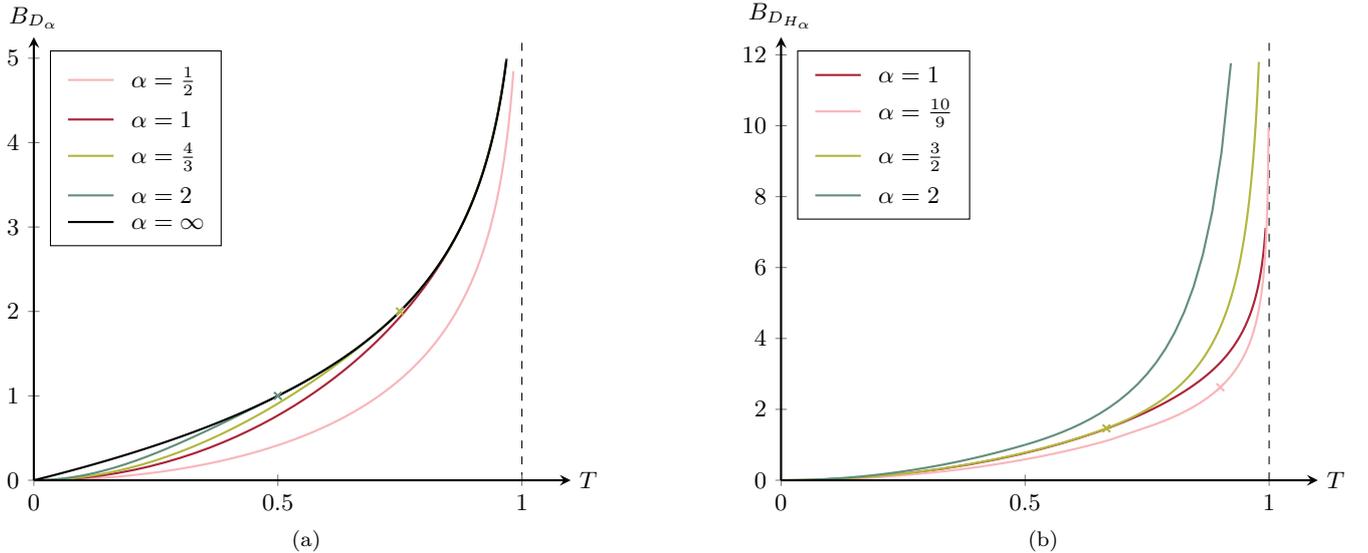
\begin{figure*}
\subfloat[]{\begin{tikzpicture}
    \begin{axis}[
            width=246pt,
            restrict y to domain=0:5,
            xmax=1.1,
            ymax=5.25,
            axis lines=left,
            axis line style=thick,
            compat=newest,
            xlabel=$T$, xlabel style={at={(1,0)}, anchor=west},
            ylabel=$B_{D_\alpha}$, ylabel style={at={(0,1)}, anchor=south, rotate=-90},
            xtick={0,0.5,1},
            axis on top
            ]
        \draw[dashed] (1,0) -- (1,20);
        % a = 1/2:
        \addplot[domain=0:1, CorinthianPink, thick, samples=400] {log2(1/(1 - x^2))} node[name=B_D]{};\label{plot:fidelity}
        % a = 1:
        \addplot[domain=0:100, samples=500, thick, color2](
          {((2^x - 1 - x * ln(2)) * (1 - 2^x + 2^x * ln(2^x)))) / ((2^x - 1)^2 * x * ln(2))},
          {(1 - 2^x + 2^x * ln(2^x)) / (2^x - 1)^2 * ln( (x) / (2^x - 1)  )/ln(2) - 2^x * (1 - 2^x + ln(2^x)) / (2^x - 1)^2 * ln((2^x * x)/(2^x - 1)) / ln(2) + ln(ln(2)) / ln(2)}
        );\label{plot:1renyi}
        % a = 4/3:
        \addplot [thick, CitronYellow] table {bound_renyi_43.csv};\label{plot:43renyi}
        % a = 2:
        \addplot[domain=0:0.5, color1, thick]{log2(1 + 4 *x^2)};
        \addplot[domain=0.5:1, color1, thick, samples=400] {log2(1/(1 - x))} node[name=B_D]{};\label{plot:2renyi}
        % a = infinity:
        \addplot[domain=0:1, black, thick, samples=600] {log2(1/(1 - x))};\label{plot:infinityrenyi}
        \addplot[mark=x, color1, thick] coordinates{(0.5,1)};
        \addplot[mark=x, CitronYellow, thick] coordinates{(0.75,2)};
        \coordinate (legend) at (axis description cs:0.03,0.97);
    \end{axis}
    \matrix [draw, matrix of nodes, anchor=north west, fill=white, column 2/.style={anchor=base west}, ampersand replacement=\&] at (legend) {
        \ref{plot:fidelity} \& $\alpha=\frac{1}{2}$ \\
        \ref{plot:1renyi} \& $\alpha=1$ \\
        \ref{plot:43renyi} \& $\alpha=\frac{4}{3}$ \\
        \ref{plot:2renyi} \& $\alpha=2$ \\
        \ref{plot:infinityrenyi} \& $\alpha=\infty$ \\
    };
\end{tikzpicture}
}\hfill
\subfloat[]{
\begin{tikzpicture}
    \begin{axis}[
            width=246pt,
            restrict y to domain=0:12,
            xmax=1.1,
            ymax=12.5,
            axis lines=left,
            axis line style=thick,
            compat=newest,
            xlabel=$T$, xlabel style={at={(1,0)}, anchor=west},
            ylabel=$B_{D_{H_\alpha}}$, ylabel style={at={(0,1)}, anchor=south, rotate=-90},
            xtick={0,0.5,1},
            axis on top
            ]
        \draw[dashed] (1,0) -- (1,20);
        % a = 1:
        \addplot[domain=0:200, samples=500, thick, color2](
          {((2^x - 1 - x * ln(2)) * (1 - 2^x + 2^x * ln(2^x)))) / ((2^x - 1)^2 * x * ln(2))},
          {(1 - 2^x + 2^x * ln(2^x)) / (2^x - 1)^2 * ln( (x) / (2^x - 1)  )/ln(2) - 2^x * (1 - 2^x + ln(2^x)) / (2^x - 1)^2 * ln((2^x * x)/(2^x - 1)) / ln(2) + ln(ln(2)) / ln(2)}
        );\label{plot:hellinger1}
        %a = 10/9:
        \addplot [thick, CorinthianPink] table {bound_hellinger_109.csv};\label{plot:hellinger109}
        \addplot[mark=x, CorinthianPink, thick] coordinates{({9/10},{-9*(1-(10)^(1/9)})};
        %a = 3/2:
        \addplot [thick, CitronYellow] table {bound_hellinger_15.csv};\label{plot:hellinger15}
        \addplot[mark=x, CitronYellow, thick] coordinates{({2/3},{-2*(1-sqrt(3))})};
        a = 2:
        \addplot[domain=0:0.5, color1, thick, name path=left bound] {4*x^2};\label{plot:2}
        \addplot[domain=0.5:{24/25}, color1, thick, name path=right bound] {x / (1-x)} node[name=B_D]{};
        \coordinate (legend) at (axis description cs:0.03,0.97);
    \end{axis}
    \matrix [draw, matrix of nodes, anchor=north west, fill=white, column 2/.style={anchor=base west}, ampersand replacement=\&] at (legend) {
    \ref{plot:hellinger1} \& $\alpha=1$ \\
    \ref{plot:hellinger109}\& $\alpha=\frac{10}{9}$ \\
    \ref{plot:hellinger15}\& $\alpha=\frac{3}{2}$ \\
    \ref{plot:2}\& $\alpha=2$ \\
    };
\end{tikzpicture}
}
\caption{\label{fig: numerical convex bounds}Convex bounds for Rényi divergences (a) and Hellinger divergences (b) each for a selection of $\alpha$. The marked points indicate where the right-side bounds and the left-side bounds meet. Remarkably, the bound for $\alpha=1$ is equal in both cases. Further, note that $B_{D_\frac{1}{2}} \leq B_{D_1} \leq B_{D_\frac{4}{3}} \leq B_{D_2} \leq B_{D_\infty}$.}
\end{figure*}

For the quantum Rényi divergence with general $\alpha$ and the Hellinger divergence with general $\alpha$, we cannot expect to find both the linear and the convex lower bounds analytically and therefore go about finding the bounds by solving both optimization steps numerically. For the first optimization, i.e.\ \eqref{eq: definition L_D}, we use scipy's minimize routine \cite{2020SciPy-NMeth} with the Nelder-Mead method, which is expected to reliably find the minimum \cite{audet_2017}. The second optimization, i.e.\ \eqref{eq: definition B_D}, is also obtained using scipy, the results of which are shown in Figure \ref{fig: numerical convex bounds}. From the points we obtain numerically, one can easily construct a convex bound in the form of a polygonal chain as depicted in Figure \ref{fig: polygonal chain}. 

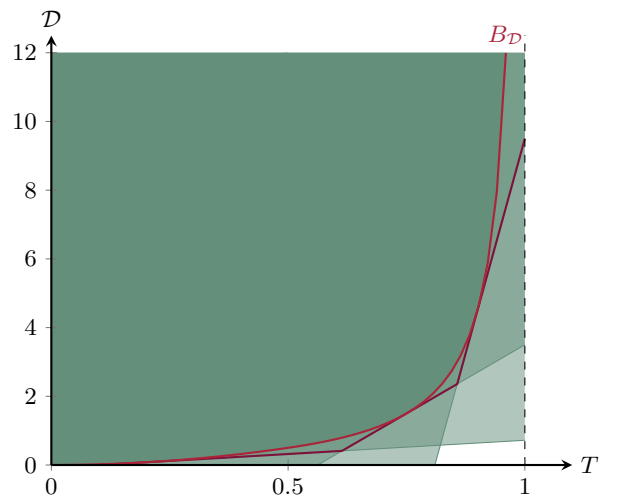
\begin{figure}
    \centering
    \begin{tikzpicture}
        \begin{axis}[
        restrict y to domain=0:12,
        xmax=1.1,
        ymax=12.5,
        axis lines=left,
        axis line style=thick,
        compat=newest,
        xlabel=$T$, xlabel style={at={(1,0)}, anchor=west},
        ylabel=$\mcd$, ylabel style={at={(0,1)}, anchor=south, rotate=-90},
        xtick={0,1/2,1},
        axis on top
        ]
        \addplot[domain=0.1:1, color1, name path=tangent 1] {0.8*(x - 0.2) + 0.08};
        \addplot[domain=0.5625:1, color1, name path=tangent 2] {8*(x - 0.75) + 1.5};
        \addplot[domain=0.81:1, color1, name path=tangent 3] {50*(x - 0.9) + 4.5};
        \addplot[black, thick, color3] table {
        0.1 0
        0.613889 0.411111
        0.857143 2.35714
        1 9.5
        };
        \draw[dashed] (1,0) -- (1,20);
        \addplot[domain=0:0.5, color2, thick, name path=left bound] {2*x^2};
        \addplot[domain=0.5:{24/25}, color2, thick, name path=right bound] {x / (2-2*x)} node[name=B_D]{};
        \path[name path=clippath] (0,0)--(0,12) -- (1,12);
        \addplot[color1] fill between [of=left bound and clippath, soft clip={domain=0:0.5}];
        \addplot[color1] fill between [of=right bound and clippath, soft clip={domain=0.49:1}];
        \addplot[color1, semitransparent] fill between [of=tangent 1 and clippath];
        \addplot[color1, semitransparent] fill between [of=tangent 2 and clippath];
        \addplot[color1, semitransparent] fill between [of=tangent 3 and clippath];
        \end{axis}
    \draw[color2] (B_D) node[anchor=south] {$B_\mcd$};
    \end{tikzpicture}
    \caption{Polygonal chain as convex lower bound. The polygonal chain acts as an approximation to the convex lower bound $B_\mcd$ and is made up of the tangents from the linear lower bound $L_\mcd$.}
    \label{fig: polygonal chain}
\end{figure}

\subsection{Other Bounds in the Literature}
We are certainly not the first to think about Pinsker-type bounds and we want to credit some of the previous results on this topic. \citeauthor{Fedotov} give a parametrized bound for the Umegaki divergence base e and an improved Pinsker inequality up (oddly) large order in \cite{Fedotov}. \citeauthor{lanier2025} obtain the same convex bound we do for $\chi^2$ in \cite{lanier2025}. Moreover, a selection of classical Pinsker-type bounds can be found in \cite{Polyanskiy_Wu_2025}. \citeauthor{Rioul_2024} gives an overview on classical Pinsker-type bounds in \cite{Rioul_2024}.

\section{summary and Outlook}\label{sec: summary and outlook}
The methods developed in this work give a systematic, surprisingly simple, and straightforward path to computing optimal Pinsker-type inequalities for arbitrary divergences. Following our two-step process, tight bounds for a given divergence can be obtained by minimizing a two-parameter objective over the unit square. While this optimization is numerically easy, deriving a closed analytic expression can range from easy to notoriously hard, and in some cases may be out of reach.

We applied our methods to a long list of divergences, listed in Table \ref{tab: results}, including many examples that are central to quantum information theory. This is, however, only a beginning. Extending the catalogue to further divergences is a natural and meaningful task for future work.
Beyond the catalogue aspect, we expect the bounds obtained here to be useful as a technical tool in many situations. Replacing an operationally meaningful divergence by an easier-to-handle quantity, such as the trace distance, is a standard step in many proofs, especially for smoothed quantities, which are known to be notoriously hard to handle in quantitative statements. 

For this reason, we are convinced that a wide range of applications of these bounds will follow.
In particular, several directions immediately seem promising:
\begin{itemize}
    \item For \emph{QKD and randomness extraction} smoothed divergences play a central role in the finite-size analysis \cite{tomamichel2011leftover}. Here, our bounds may simplify the estimation of attainable rates. Especially since the trace distance itself can be lower bounded by observed statistical data. 
    
    \item For \emph{efficient quantum Gibbs sampling}, recent work \cite{chen2023efficient} suggests that the mixing-time factor is a tight bottleneck for the overall runtime. Bounds on these quantities can be obtained via Logarithmic Sobolev inequalities, which in turn can be made numerically tangible by applying a Pinsker-type inequality. 

    \item \emph{Entropic uncertainty relations} are a central tool used throughout quantum physics \cite{schwonnek2018additivity,schwonnek2018uncertainty,rotundo2024entropic}. Despite their operational importance, their numerical estimation can often be a difficult problem. Replacing entropic quantities with trace distances via optimal bounds is hence a potentially fruitful route to take. 
\end{itemize}

\subsection*{Acknowledgements}
G.K.\ acknowledges support from the Excellence Cluster -- Matter and Light for Quantum Computing (ML4Q) and funding by the European Research Council (ERC Grant Agreement No.\ 948139). R.S.\ is supported  by the DFG under Germany's Excellence Strategy - EXC-2123 QuantumFrontiers - 390837967 and  SFB 1227 (DQ-mat), the Quantum Valley Lower Saxony, and the BMBF projects CBQD, SEQUIN, Quics and ATIQ.

%%%%% Bibliography %%%%%
\bibliography{bibliography.bib}

%apsrev4-2.bst 2019-01-14 (MD) hand-edited version of apsrev4-1.bst
%Control: key (0)
%Control: author (72) initials jnrlst
%Control: editor formatted (1) identically to author
%Control: production of article title (-1) disabled
%Control: page (0) single
%Control: year (1) truncated
%Control: production of eprint (0) enabled
\begin{thebibliography}{35}%
\makeatletter
\providecommand \@ifxundefined [1]{%
 \@ifx{#1\undefined}
}%
\providecommand \@ifnum [1]{%
 \ifnum #1\expandafter \@firstoftwo
 \else \expandafter \@secondoftwo
 \fi
}%
\providecommand \@ifx [1]{%
 \ifx #1\expandafter \@firstoftwo
 \else \expandafter \@secondoftwo
 \fi
}%
\providecommand \natexlab [1]{#1}%
\providecommand \enquote  [1]{``#1''}%
\providecommand \bibnamefont  [1]{#1}%
\providecommand \bibfnamefont [1]{#1}%
\providecommand \citenamefont [1]{#1}%
\providecommand \href@noop [0]{\@secondoftwo}%
\providecommand \href [0]{\begingroup \@sanitize@url \@href}%
\providecommand \@href[1]{\@@startlink{#1}\@@href}%
\providecommand \@@href[1]{\endgroup#1\@@endlink}%
\providecommand \@sanitize@url [0]{\catcode `\\12\catcode `\$12\catcode `\&12\catcode `\#12\catcode `\^12\catcode `\_12\catcode `\%12\relax}%
\providecommand \@@startlink[1]{}%
\providecommand \@@endlink[0]{}%
\providecommand \url  [0]{\begingroup\@sanitize@url \@url }%
\providecommand \@url [1]{\endgroup\@href {#1}{\urlprefix }}%
\providecommand \urlprefix  [0]{URL }%
\providecommand \Eprint [0]{\href }%
\providecommand \doibase [0]{https://doi.org/}%
\providecommand \selectlanguage [0]{\@gobble}%
\providecommand \bibinfo  [0]{\@secondoftwo}%
\providecommand \bibfield  [0]{\@secondoftwo}%
\providecommand \translation [1]{[#1]}%
\providecommand \BibitemOpen [0]{}%
\providecommand \bibitemStop [0]{}%
\providecommand \bibitemNoStop [0]{.\EOS\space}%
\providecommand \EOS [0]{\spacefactor3000\relax}%
\providecommand \BibitemShut  [1]{\csname bibitem#1\endcsname}%
\let\auto@bib@innerbib\@empty
%</preamble>
\bibitem [{\citenamefont {Petz}(1986)}]{Petz1986}%
  \BibitemOpen
  \bibfield  {author} {\bibinfo {author} {\bibfnamefont {D.}~\bibnamefont {Petz}},\ }\href {https://doi.org/10.1016/0034-4877(86)90067-4} {\bibfield  {journal} {\bibinfo  {journal} {Reports on Mathematical Physics}\ }\textbf {\bibinfo {volume} {23}},\ \bibinfo {pages} {57–65} (\bibinfo {year} {1986})}\BibitemShut {NoStop}%
\bibitem [{\citenamefont {Petz}(2010)}]{Petz2010}%
  \BibitemOpen
  \bibfield  {author} {\bibinfo {author} {\bibfnamefont {D.}~\bibnamefont {Petz}},\ }\href {https://doi.org/10.3390/e12030304} {\bibfield  {journal} {\bibinfo  {journal} {Entropy}\ }\textbf {\bibinfo {volume} {12}},\ \bibinfo {pages} {304–325} (\bibinfo {year} {2010})}\BibitemShut {NoStop}%
\bibitem [{\citenamefont {Wilde}\ \emph {et~al.}(2014)\citenamefont {Wilde}, \citenamefont {Winter},\ and\ \citenamefont {Yang}}]{Wilde_2014}%
  \BibitemOpen
  \bibfield  {author} {\bibinfo {author} {\bibfnamefont {M.~M.}\ \bibnamefont {Wilde}}, \bibinfo {author} {\bibfnamefont {A.}~\bibnamefont {Winter}},\ and\ \bibinfo {author} {\bibfnamefont {D.}~\bibnamefont {Yang}},\ }\href {https://doi.org/10.1007/s00220-014-2122-x} {\bibfield  {journal} {\bibinfo  {journal} {Communications in Mathematical Physics}\ }\textbf {\bibinfo {volume} {331}},\ \bibinfo {pages} {593–622} (\bibinfo {year} {2014})}\BibitemShut {NoStop}%
\bibitem [{\citenamefont {M\"{u}ller-Lennert}\ \emph {et~al.}(2013)\citenamefont {M\"{u}ller-Lennert}, \citenamefont {Dupuis}, \citenamefont {Szehr}, \citenamefont {Fehr},\ and\ \citenamefont {Tomamichel}}]{MllerLennert2013}%
  \BibitemOpen
  \bibfield  {author} {\bibinfo {author} {\bibfnamefont {M.}~\bibnamefont {M\"{u}ller-Lennert}}, \bibinfo {author} {\bibfnamefont {F.}~\bibnamefont {Dupuis}}, \bibinfo {author} {\bibfnamefont {O.}~\bibnamefont {Szehr}}, \bibinfo {author} {\bibfnamefont {S.}~\bibnamefont {Fehr}},\ and\ \bibinfo {author} {\bibfnamefont {M.}~\bibnamefont {Tomamichel}},\ }\bibfield  {journal} {\bibinfo  {journal} {Journal of Mathematical Physics}\ }\textbf {\bibinfo {volume} {54}},\ \href {https://doi.org/10.1063/1.4838856} {10.1063/1.4838856} (\bibinfo {year} {2013})\BibitemShut {NoStop}%
\bibitem [{\citenamefont {Tomamichel}(2012)}]{TomamichelPhDthesis}%
  \BibitemOpen
  \bibfield  {author} {\bibinfo {author} {\bibfnamefont {M.}~\bibnamefont {Tomamichel}},\ }\href {https://doi.org/10.48550/ARXIV.1203.2142} {\bibinfo {title} {A framework for non-asymptotic quantum information theory}} (\bibinfo {year} {2012})\BibitemShut {NoStop}%
\bibitem [{\citenamefont {Stinespring}(1955)}]{Stinespring1955}%
  \BibitemOpen
  \bibfield  {author} {\bibinfo {author} {\bibfnamefont {W.~F.}\ \bibnamefont {Stinespring}},\ }\href {https://doi.org/10.2307/2032342} {\bibfield  {journal} {\bibinfo  {journal} {Proceedings of the American Mathematical Society}\ }\textbf {\bibinfo {volume} {6}},\ \bibinfo {pages} {211} (\bibinfo {year} {1955})}\BibitemShut {NoStop}%
\bibitem [{\citenamefont {Watrous}(2018)}]{Watrous_2018}%
  \BibitemOpen
  \bibfield  {author} {\bibinfo {author} {\bibfnamefont {J.}~\bibnamefont {Watrous}},\ }\href {https://doi.org/10.1017/9781316848142} {\emph {\bibinfo {title} {The Theory of Quantum Information}}}\ (\bibinfo  {publisher} {Cambridge University Press},\ \bibinfo {year} {2018})\BibitemShut {NoStop}%
\bibitem [{\citenamefont {Hiai}\ and\ \citenamefont {Petz}(1991)}]{Hiai1991}%
  \BibitemOpen
  \bibfield  {author} {\bibinfo {author} {\bibfnamefont {F.}~\bibnamefont {Hiai}}\ and\ \bibinfo {author} {\bibfnamefont {D.}~\bibnamefont {Petz}},\ }\href {https://doi.org/10.1007/bf02100287} {\bibfield  {journal} {\bibinfo  {journal} {Communications in Mathematical Physics}\ }\textbf {\bibinfo {volume} {143}},\ \bibinfo {pages} {99–114} (\bibinfo {year} {1991})}\BibitemShut {NoStop}%
\bibitem [{\citenamefont {Bluhm}\ \emph {et~al.}(2023)\citenamefont {Bluhm}, \citenamefont {Capel}, \citenamefont {Gondolf},\ and\ \citenamefont {P{\'e}rez~Hern{\'a}ndez}}]{Bluhm_2023}%
  \BibitemOpen
  \bibfield  {author} {\bibinfo {author} {\bibfnamefont {A.}~\bibnamefont {Bluhm}}, \bibinfo {author} {\bibfnamefont {{\'A}.}~\bibnamefont {Capel}}, \bibinfo {author} {\bibfnamefont {P.}~\bibnamefont {Gondolf}},\ and\ \bibinfo {author} {\bibfnamefont {A.}~\bibnamefont {P{\'e}rez~Hern{\'a}ndez}},\ }in\ \href {https://doi.org/10.1109/ISIT54713.2023.10206734} {\emph {\bibinfo {booktitle} {2023 IEEE International Symposium on Information Theory (ISIT)}}}\ (\bibinfo  {publisher} {IEEE},\ \bibinfo {year} {2023})\ pp.\ \bibinfo {pages} {162--167}\BibitemShut {NoStop}%
\bibitem [{\citenamefont {Cover}\ and\ \citenamefont {Thomas}(2005)}]{Cover_2006}%
  \BibitemOpen
  \bibfield  {author} {\bibinfo {author} {\bibfnamefont {T.~M.}\ \bibnamefont {Cover}}\ and\ \bibinfo {author} {\bibfnamefont {J.~A.}\ \bibnamefont {Thomas}},\ }\href {https://doi.org/10.1002/047174882x} {\emph {\bibinfo {title} {Elements of Information Theory}}}\ (\bibinfo  {publisher} {Wiley},\ \bibinfo {year} {2005})\BibitemShut {NoStop}%
\bibitem [{\citenamefont {Müller-Lennert}\ \emph {et~al.}(2013)\citenamefont {Müller-Lennert}, \citenamefont {Dupuis}, \citenamefont {Szehr}, \citenamefont {Fehr},\ and\ \citenamefont {Tomamichel}}]{TomamichelDupuis2013}%
  \BibitemOpen
  \bibfield  {author} {\bibinfo {author} {\bibfnamefont {M.}~\bibnamefont {Müller-Lennert}}, \bibinfo {author} {\bibfnamefont {F.}~\bibnamefont {Dupuis}}, \bibinfo {author} {\bibfnamefont {O.}~\bibnamefont {Szehr}}, \bibinfo {author} {\bibfnamefont {S.}~\bibnamefont {Fehr}},\ and\ \bibinfo {author} {\bibfnamefont {M.}~\bibnamefont {Tomamichel}},\ }\bibfield  {journal} {\bibinfo  {journal} {Journal of Mathematical Physics}\ }\textbf {\bibinfo {volume} {54}},\ \href {https://doi.org/10.1063/1.4838856} {10.1063/1.4838856} (\bibinfo {year} {2013})\BibitemShut {NoStop}%
\bibitem [{\citenamefont {Regula}\ \emph {et~al.}(2025)\citenamefont {Regula}, \citenamefont {Lami},\ and\ \citenamefont {Datta}}]{Regula_2025}%
  \BibitemOpen
  \bibfield  {author} {\bibinfo {author} {\bibfnamefont {B.}~\bibnamefont {Regula}}, \bibinfo {author} {\bibfnamefont {L.}~\bibnamefont {Lami}},\ and\ \bibinfo {author} {\bibfnamefont {N.}~\bibnamefont {Datta}},\ }\href {https://arxiv.org/abs/2501.12447} {\bibinfo {title} {Tight relations and equivalences between smooth relative entropies}} (\bibinfo {year} {2025}),\ \Eprint {https://arxiv.org/abs/2501.12447} {arXiv:2501.12447 [quant-ph]} \BibitemShut {NoStop}%
\bibitem [{\citenamefont {Hirche}\ and\ \citenamefont {Tomamichel}(2024)}]{Hirche_2024}%
  \BibitemOpen
  \bibfield  {author} {\bibinfo {author} {\bibfnamefont {C.}~\bibnamefont {Hirche}}\ and\ \bibinfo {author} {\bibfnamefont {M.}~\bibnamefont {Tomamichel}},\ }\bibfield  {journal} {\bibinfo  {journal} {Communications in Mathematical Physics}\ }\textbf {\bibinfo {volume} {405}},\ \href {https://doi.org/10.1007/s00220-024-05087-3} {10.1007/s00220-024-05087-3} (\bibinfo {year} {2024})\BibitemShut {NoStop}%
\bibitem [{\citenamefont {Beigi}\ \emph {et~al.}(2025)\citenamefont {Beigi}, \citenamefont {Hirche},\ and\ \citenamefont {Tomamichel}}]{Beigi_2025}%
  \BibitemOpen
  \bibfield  {author} {\bibinfo {author} {\bibfnamefont {S.}~\bibnamefont {Beigi}}, \bibinfo {author} {\bibfnamefont {C.}~\bibnamefont {Hirche}},\ and\ \bibinfo {author} {\bibfnamefont {M.}~\bibnamefont {Tomamichel}},\ }\href {https://arxiv.org/abs/2501.03799} {\bibinfo {title} {Some properties and applications of the new quantum $f$-divergences}} (\bibinfo {year} {2025}),\ \Eprint {https://arxiv.org/abs/2501.03799} {arXiv:2501.03799 [quant-ph]} \BibitemShut {NoStop}%
\bibitem [{\citenamefont {Tomamichel}(2016)}]{Tomamichel_2016}%
  \BibitemOpen
  \bibfield  {author} {\bibinfo {author} {\bibfnamefont {M.}~\bibnamefont {Tomamichel}},\ }\href {https://doi.org/10.1007/978-3-319-21891-5} {\emph {\bibinfo {title} {Quantum Information Processing with Finite Resources}}}\ (\bibinfo  {publisher} {Springer International Publishing},\ \bibinfo {year} {2016})\BibitemShut {NoStop}%
\bibitem [{\citenamefont {Ehrgott}(2000)}]{ehrgott_2000}%
  \BibitemOpen
  \bibfield  {author} {\bibinfo {author} {\bibfnamefont {M.}~\bibnamefont {Ehrgott}},\ }\href {https://doi.org/10.1007/978-3-662-22199-0} {\emph {\bibinfo {title} {Multicriteria Optimization}}},\ \bibinfo {edition} {1st}\ ed.,\ \bibinfo {series} {Lecture Notes in Economics and Mathematical Systems}, Vol.\ \bibinfo {volume} {491}\ (\bibinfo  {publisher} {Springer Berlin Heidelberg},\ \bibinfo {year} {2000})\BibitemShut {NoStop}%
\bibitem [{\citenamefont {Polyanskiy}\ and\ \citenamefont {Wu}(2025)}]{Polyanskiy_Wu_2025}%
  \BibitemOpen
  \bibfield  {author} {\bibinfo {author} {\bibfnamefont {Y.}~\bibnamefont {Polyanskiy}}\ and\ \bibinfo {author} {\bibfnamefont {Y.}~\bibnamefont {Wu}},\ }\href {https://doi.org/10.1017/9781108966351} {\emph {\bibinfo {title} {Information Theory: From Coding to Learning}}}\ (\bibinfo  {publisher} {Cambridge University Press},\ \bibinfo {year} {2025})\BibitemShut {NoStop}%
\bibitem [{\citenamefont {Wilde}(2016)}]{Wilde_2016}%
  \BibitemOpen
  \bibfield  {author} {\bibinfo {author} {\bibfnamefont {M.~M.}\ \bibnamefont {Wilde}},\ }\href {https://doi.org/10.1017/9781316809976} {\emph {\bibinfo {title} {Quantum Information Theory}}}\ (\bibinfo  {publisher} {Cambridge University Press},\ \bibinfo {year} {2016})\BibitemShut {NoStop}%
\bibitem [{\citenamefont {Hiai}\ and\ \citenamefont {Mosonyi}(2017)}]{Hiai_2017}%
  \BibitemOpen
  \bibfield  {author} {\bibinfo {author} {\bibfnamefont {F.}~\bibnamefont {Hiai}}\ and\ \bibinfo {author} {\bibfnamefont {M.}~\bibnamefont {Mosonyi}},\ }\href {https://doi.org/10.1142/s0129055x17500234} {\bibfield  {journal} {\bibinfo  {journal} {Reviews in Mathematical Physics}\ }\textbf {\bibinfo {volume} {29}},\ \bibinfo {pages} {1750023} (\bibinfo {year} {2017})}\BibitemShut {NoStop}%
\bibitem [{\citenamefont {Wilde}(2018)}]{Wilde_2018}%
  \BibitemOpen
  \bibfield  {author} {\bibinfo {author} {\bibfnamefont {M.~M.}\ \bibnamefont {Wilde}},\ }\href {https://doi.org/10.1088/1751-8121/aad5a1} {\bibfield  {journal} {\bibinfo  {journal} {Journal of Physics A: Mathematical and Theoretical}\ }\textbf {\bibinfo {volume} {51}},\ \bibinfo {pages} {374002} (\bibinfo {year} {2018})}\BibitemShut {NoStop}%
\bibitem [{\citenamefont {Audenaert}\ and\ \citenamefont {Datta}(2015)}]{Audenaert_2015}%
  \BibitemOpen
  \bibfield  {author} {\bibinfo {author} {\bibfnamefont {K.~M.~R.}\ \bibnamefont {Audenaert}}\ and\ \bibinfo {author} {\bibfnamefont {N.}~\bibnamefont {Datta}},\ }\bibfield  {journal} {\bibinfo  {journal} {Journal of Mathematical Physics}\ }\textbf {\bibinfo {volume} {56}},\ \href {https://doi.org/10.1063/1.4906367} {10.1063/1.4906367} (\bibinfo {year} {2015})\BibitemShut {NoStop}%
\bibitem [{\citenamefont {Anshu}\ \emph {et~al.}(2019)\citenamefont {Anshu}, \citenamefont {Berta}, \citenamefont {Jain},\ and\ \citenamefont {Tomamichel}}]{Anshu_2019}%
  \BibitemOpen
  \bibfield  {author} {\bibinfo {author} {\bibfnamefont {A.}~\bibnamefont {Anshu}}, \bibinfo {author} {\bibfnamefont {M.}~\bibnamefont {Berta}}, \bibinfo {author} {\bibfnamefont {R.}~\bibnamefont {Jain}},\ and\ \bibinfo {author} {\bibfnamefont {M.}~\bibnamefont {Tomamichel}},\ }\bibfield  {journal} {\bibinfo  {journal} {Journal of Mathematical Physics}\ }\textbf {\bibinfo {volume} {60}},\ \href {https://doi.org/10.1063/1.5126723} {10.1063/1.5126723} (\bibinfo {year} {2019})\BibitemShut {NoStop}%
\bibitem [{\citenamefont {Tsybakov}(2009)}]{Tsybakov_2009}%
  \BibitemOpen
  \bibfield  {author} {\bibinfo {author} {\bibfnamefont {A.~B.}\ \bibnamefont {Tsybakov}},\ }\href {https://doi.org/10.1007/978-0-387-79052-7} {\emph {\bibinfo {title} {Introduction to Nonparametric Estimation}}}\ (\bibinfo  {publisher} {Springer New York},\ \bibinfo {address} {New York, NY},\ \bibinfo {year} {2009})\BibitemShut {NoStop}%
\bibitem [{\citenamefont {Bell}(1934)}]{bell1934exponential}%
  \BibitemOpen
  \bibfield  {author} {\bibinfo {author} {\bibfnamefont {E.~T.}\ \bibnamefont {Bell}},\ }\href@noop {} {\bibfield  {journal} {\bibinfo  {journal} {Annals of Mathematics}\ }\textbf {\bibinfo {volume} {35}},\ \bibinfo {pages} {258} (\bibinfo {year} {1934})}\BibitemShut {NoStop}%
\bibitem [{\citenamefont {Fedotov}\ \emph {et~al.}(2003)\citenamefont {Fedotov}, \citenamefont {Harremo\"es},\ and\ \citenamefont {Tops\oe}}]{Fedotov}%
  \BibitemOpen
  \bibfield  {author} {\bibinfo {author} {\bibfnamefont {A.}~\bibnamefont {Fedotov}}, \bibinfo {author} {\bibfnamefont {P.}~\bibnamefont {Harremo\"es}},\ and\ \bibinfo {author} {\bibfnamefont {F.}~\bibnamefont {Tops\oe}},\ }\href {https://doi.org/10.1109/TIT.2003.811927} {\bibfield  {journal} {\bibinfo  {journal} {IEEE Transactions on Information Theory}\ }\textbf {\bibinfo {volume} {49}},\ \bibinfo {pages} {1491} (\bibinfo {year} {2003})}\BibitemShut {NoStop}%
\bibitem [{\citenamefont {Virtanen}\ \emph {et~al.}(2020)\citenamefont {Virtanen}, \citenamefont {Gommers}, \citenamefont {Oliphant}, \citenamefont {Haberland}, \citenamefont {Reddy}, \citenamefont {Cournapeau}, \citenamefont {Burovski}, \citenamefont {Peterson}, \citenamefont {Weckesser}, \citenamefont {Bright}, \citenamefont {{van der Walt}}, \citenamefont {Brett}, \citenamefont {Wilson}, \citenamefont {Millman}, \citenamefont {Mayorov}, \citenamefont {Nelson}, \citenamefont {Jones}, \citenamefont {Kern}, \citenamefont {Larson}, \citenamefont {Carey}, \citenamefont {Polat}, \citenamefont {Feng}, \citenamefont {Moore}, \citenamefont {{VanderPlas}}, \citenamefont {Laxalde}, \citenamefont {Perktold}, \citenamefont {Cimrman}, \citenamefont {Henriksen}, \citenamefont {Quintero}, \citenamefont {Harris}, \citenamefont {Archibald}, \citenamefont {Ribeiro}, \citenamefont {Pedregosa}, \citenamefont {{van Mulbregt}},\ and\ \citenamefont {{SciPy 1.0 Contributors}}}]{2020SciPy-NMeth}%
  \BibitemOpen
  \bibfield  {author} {\bibinfo {author} {\bibfnamefont {P.}~\bibnamefont {Virtanen}}, \bibinfo {author} {\bibfnamefont {R.}~\bibnamefont {Gommers}}, \bibinfo {author} {\bibfnamefont {T.~E.}\ \bibnamefont {Oliphant}}, \bibinfo {author} {\bibfnamefont {M.}~\bibnamefont {Haberland}}, \bibinfo {author} {\bibfnamefont {T.}~\bibnamefont {Reddy}}, \bibinfo {author} {\bibfnamefont {D.}~\bibnamefont {Cournapeau}}, \bibinfo {author} {\bibfnamefont {E.}~\bibnamefont {Burovski}}, \bibinfo {author} {\bibfnamefont {P.}~\bibnamefont {Peterson}}, \bibinfo {author} {\bibfnamefont {W.}~\bibnamefont {Weckesser}}, \bibinfo {author} {\bibfnamefont {J.}~\bibnamefont {Bright}}, \bibinfo {author} {\bibfnamefont {S.~J.}\ \bibnamefont {{van der Walt}}}, \bibinfo {author} {\bibfnamefont {M.}~\bibnamefont {Brett}}, \bibinfo {author} {\bibfnamefont {J.}~\bibnamefont {Wilson}}, \bibinfo {author} {\bibfnamefont {K.~J.}\ \bibnamefont {Millman}}, \bibinfo {author} {\bibfnamefont {N.}~\bibnamefont {Mayorov}}, \bibinfo {author} {\bibfnamefont
  {A.~R.~J.}\ \bibnamefont {Nelson}}, \bibinfo {author} {\bibfnamefont {E.}~\bibnamefont {Jones}}, \bibinfo {author} {\bibfnamefont {R.}~\bibnamefont {Kern}}, \bibinfo {author} {\bibfnamefont {E.}~\bibnamefont {Larson}}, \bibinfo {author} {\bibfnamefont {C.~J.}\ \bibnamefont {Carey}}, \bibinfo {author} {\bibfnamefont {{\.I}.}~\bibnamefont {Polat}}, \bibinfo {author} {\bibfnamefont {Y.}~\bibnamefont {Feng}}, \bibinfo {author} {\bibfnamefont {E.~W.}\ \bibnamefont {Moore}}, \bibinfo {author} {\bibfnamefont {J.}~\bibnamefont {{VanderPlas}}}, \bibinfo {author} {\bibfnamefont {D.}~\bibnamefont {Laxalde}}, \bibinfo {author} {\bibfnamefont {J.}~\bibnamefont {Perktold}}, \bibinfo {author} {\bibfnamefont {R.}~\bibnamefont {Cimrman}}, \bibinfo {author} {\bibfnamefont {I.}~\bibnamefont {Henriksen}}, \bibinfo {author} {\bibfnamefont {E.~A.}\ \bibnamefont {Quintero}}, \bibinfo {author} {\bibfnamefont {C.~R.}\ \bibnamefont {Harris}}, \bibinfo {author} {\bibfnamefont {A.~M.}\ \bibnamefont {Archibald}}, \bibinfo {author}
  {\bibfnamefont {A.~H.}\ \bibnamefont {Ribeiro}}, \bibinfo {author} {\bibfnamefont {F.}~\bibnamefont {Pedregosa}}, \bibinfo {author} {\bibfnamefont {P.}~\bibnamefont {{van Mulbregt}}},\ and\ \bibinfo {author} {\bibnamefont {{SciPy 1.0 Contributors}}},\ }\href {https://doi.org/10.1038/s41592-019-0686-2} {\bibfield  {journal} {\bibinfo  {journal} {Nature Methods}\ }\textbf {\bibinfo {volume} {17}},\ \bibinfo {pages} {261} (\bibinfo {year} {2020})}\BibitemShut {NoStop}%
\bibitem [{\citenamefont {Audet}\ and\ \citenamefont {Hare}(2017)}]{audet_2017}%
  \BibitemOpen
  \bibfield  {author} {\bibinfo {author} {\bibfnamefont {C.}~\bibnamefont {Audet}}\ and\ \bibinfo {author} {\bibfnamefont {W.}~\bibnamefont {Hare}},\ }\href {https://doi.org/10.1007/978-3-319-68913-5} {\emph {\bibinfo {title} {Derivative-Free and Blackbox Optimization}}},\ \bibinfo {edition} {1st}\ ed.,\ Springer Series in Operations Research and Financial Engineering\ (\bibinfo  {publisher} {Springer, Cham},\ \bibinfo {year} {2017})\BibitemShut {NoStop}%
\bibitem [{\citenamefont {Lanier}\ \emph {et~al.}(2025)\citenamefont {Lanier}, \citenamefont {Béguinot},\ and\ \citenamefont {Rioul}}]{lanier2025}%
  \BibitemOpen
  \bibfield  {author} {\bibinfo {author} {\bibfnamefont {D.}~\bibnamefont {Lanier}}, \bibinfo {author} {\bibfnamefont {J.}~\bibnamefont {Béguinot}},\ and\ \bibinfo {author} {\bibfnamefont {O.}~\bibnamefont {Rioul}},\ }\href {https://arxiv.org/abs/2501.14340} {\bibinfo {title} {From classical to quantum: Explicit classical distributions achieving maximal quantum $f$-divergence}} (\bibinfo {year} {2025}),\ \Eprint {https://arxiv.org/abs/2501.14340} {arXiv:2501.14340 [quant-ph]} \BibitemShut {NoStop}%
\bibitem [{\citenamefont {Rioul}(2024)}]{Rioul_2024}%
  \BibitemOpen
  \bibfield  {author} {\bibinfo {author} {\bibfnamefont {O.}~\bibnamefont {Rioul}},\ }\href {https://doi.org/10.1007/s41884-024-00138-z} {\bibfield  {journal} {\bibinfo  {journal} {Information Geometry}\ }\textbf {\bibinfo {volume} {7}},\ \bibinfo {pages} {737} (\bibinfo {year} {2024})}\BibitemShut {NoStop}%
\bibitem [{\citenamefont {Tomamichel}\ \emph {et~al.}(2011)\citenamefont {Tomamichel}, \citenamefont {Schaffner}, \citenamefont {Smith},\ and\ \citenamefont {Renner}}]{tomamichel2011leftover}%
  \BibitemOpen
  \bibfield  {author} {\bibinfo {author} {\bibfnamefont {M.}~\bibnamefont {Tomamichel}}, \bibinfo {author} {\bibfnamefont {C.}~\bibnamefont {Schaffner}}, \bibinfo {author} {\bibfnamefont {A.}~\bibnamefont {Smith}},\ and\ \bibinfo {author} {\bibfnamefont {R.}~\bibnamefont {Renner}},\ }\href@noop {} {\bibfield  {journal} {\bibinfo  {journal} {IEEE Transactions on Information Theory}\ }\textbf {\bibinfo {volume} {57}},\ \bibinfo {pages} {5524} (\bibinfo {year} {2011})}\BibitemShut {NoStop}%
\bibitem [{\citenamefont {Chen}\ \emph {et~al.}(2023)\citenamefont {Chen}, \citenamefont {Kastoryano},\ and\ \citenamefont {Gily{\'e}n}}]{chen2023efficient}%
  \BibitemOpen
  \bibfield  {author} {\bibinfo {author} {\bibfnamefont {C.-F.}\ \bibnamefont {Chen}}, \bibinfo {author} {\bibfnamefont {M.~J.}\ \bibnamefont {Kastoryano}},\ and\ \bibinfo {author} {\bibfnamefont {A.}~\bibnamefont {Gily{\'e}n}},\ }\href@noop {} {\bibfield  {journal} {\bibinfo  {journal} {arXiv preprint arXiv:2311.09207}\ } (\bibinfo {year} {2023})}\BibitemShut {NoStop}%
\bibitem [{\citenamefont {Schwonnek}(2018{\natexlab{a}})}]{schwonnek2018additivity}%
  \BibitemOpen
  \bibfield  {author} {\bibinfo {author} {\bibfnamefont {R.}~\bibnamefont {Schwonnek}},\ }\href@noop {} {\bibfield  {journal} {\bibinfo  {journal} {Quantum}\ }\textbf {\bibinfo {volume} {2}},\ \bibinfo {pages} {59} (\bibinfo {year} {2018}{\natexlab{a}})}\BibitemShut {NoStop}%
\bibitem [{\citenamefont {Schwonnek}(2018{\natexlab{b}})}]{schwonnek2018uncertainty}%
  \BibitemOpen
  \bibfield  {author} {\bibinfo {author} {\bibfnamefont {R.}~\bibnamefont {Schwonnek}},\ }\href@noop {} {\bibinfo {title} {Uncertainty relations in quantum theory}} (\bibinfo {year} {2018}{\natexlab{b}})\BibitemShut {NoStop}%
\bibitem [{\citenamefont {Rotundo}\ and\ \citenamefont {Schwonnek}(2024)}]{rotundo2024entropic}%
  \BibitemOpen
  \bibfield  {author} {\bibinfo {author} {\bibfnamefont {A.~F.}\ \bibnamefont {Rotundo}}\ and\ \bibinfo {author} {\bibfnamefont {R.}~\bibnamefont {Schwonnek}},\ }\href@noop {} {\bibfield  {journal} {\bibinfo  {journal} {Physical Review Research}\ }\textbf {\bibinfo {volume} {6}},\ \bibinfo {pages} {033043} (\bibinfo {year} {2024})}\BibitemShut {NoStop}%
\bibitem [{\citenamefont {Boyd}\ and\ \citenamefont {Vandenberghe}(2004)}]{Boyd_2004}%
  \BibitemOpen
  \bibfield  {author} {\bibinfo {author} {\bibfnamefont {S.}~\bibnamefont {Boyd}}\ and\ \bibinfo {author} {\bibfnamefont {L.}~\bibnamefont {Vandenberghe}},\ }\href {https://web.stanford.edu/~boyd/cvxbook/} {\emph {\bibinfo {title} {Convex Optimization}}}\ (\bibinfo  {publisher} {Cambridge University Press},\ \bibinfo {address} {Cambridge},\ \bibinfo {year} {2004})\BibitemShut {NoStop}%
\end{thebibliography}%

%%%%% Appendix %%%%%
\appendix

\section{Proof of Theorem \ref{thm: main theorem}}\label{Appendix: Proof of main theorem}
\begin{proof}
We begin by proving that 
\begin{align}
   &L_{\mcd}(\lambda)=\inf_{r,s} \sof{\mcd_\textup{bin}(\midmid{r}{s})- \lambda (r-s) }\quad \\ &\operatorname{s.th. } 0\leq s\leq r\leq 1.\nonumber
\end{align}

Let $\rho, \sigma \in \mathcal{S}(\mathcal{H})$ be two states. Through a suitable choice of basis we can bring $\rho - \sigma$ into the block form
    \begin{equation}
        \rho - \sigma = \of*{\begin{array}{c|c}
            (\rho - \sigma)^+ & 0 \\
            \hline
            0 & (\rho - \sigma)^-
        \end{array}},
    \end{equation}
    where $( \rho - \sigma)^\pm$ denotes the positive and negative part of $\rho - \sigma$, respectively. Without loss of generality we include the zero part in $( \rho - \sigma)^+$.

    For every pair of $\rho$ and $\sigma$ we define the \emph{classicalization channel} $\mathcal{C}_{\rho, \sigma}$ to be the channel that takes the trace on each of the two blocks on the main diagonal of $\rho - \sigma$, i. e.
    \begin{equation}
        \mathcal{C}_{\rho, \sigma}(\rho - \sigma) \coloneqq \begin{pmatrix}
            \tr((\rho - \sigma)^+) & 0 \\
            0 & \tr((\rho - \sigma)^-)
        \end{pmatrix} .
    \end{equation}
    Effectively, $\mathcal{C}_{\rho, \sigma}$ reduces $\rho - \sigma$ to a two-dimensional classical form, which we may write as
    \begin{align}
        T_{\rho, \sigma}(\rho - \sigma) &= \begin{pmatrix}
            r - s & 0 \\
            0 & (1 - r) - (1 - s)
        \end{pmatrix} \\
        &= \begin{pmatrix}
            r - s & 0 \\
            0 & s-r
        \end{pmatrix},
    \end{align}
    where $0 \leq s \leq r \leq 1$ because the upper left part is positive by our choice of basis and the trace of $\rho - \sigma$ vanishes. For $\rho$ and $\sigma$ separately we may write
    \begin{equation}
        \mathcal{C}_{\rho, \sigma} (\rho) = \begin{pmatrix}
            r & 0 \\
            0 & 1-r
        \end{pmatrix}, \quad
        \mathcal{C}_{\rho, \sigma} (\sigma) = \begin{pmatrix}
            s & 0 \\
            0 & 1-s
        \end{pmatrix}.
    \end{equation}
    This channel conserves the trace distance, 
    \begin{equation}\label{eq: conservation trace distance}
        T(\rho, \sigma ) = T (\mathcal{C}_{\rho, \sigma} (\rho), \mathcal{C}_{\rho, \sigma} (\sigma) ) ,
    \end{equation}
    because $\Vert \rho - \sigma\Vert_1$ is the sum of the absolute value of all eigenvalues of $\rho - \sigma$, which is the same as the sum $\abs{\tr((\rho - \sigma)^+)} + \abs{\tr((\rho - \sigma)^-)} = \Vert \mathcal{C}_{\rho, \sigma} (\rho - \sigma)\Vert_1$. 
    
    Using the conservation of the trace distance \eqref{eq: conservation trace distance} and the data processing inequality \eqref{eq: data processing} for $\mathcal{C}_{\rho, \sigma}$, we can simplify the lower linear bound to the desired form:
    \begin{align}
        L_{\mcd}(\lambda) &= \inf_{\rho, \sigma} \sof{\mcd(\midmid{\rho}{\sigma}) - \lambda T (\rho, \sigma)} \\
        &= \inf_{\rho, \sigma} \sof{\mcd(\midmid{\mathcal{C}_{\rho, \sigma}(\rho)}{\mathcal{C}_{\rho, \sigma}(\sigma)}) \\
        &\qquad \quad- \lambda T (\mathcal{C}_{\rho, \sigma}(\rho), \mathcal{C}_{\rho, \sigma}(\sigma) )}\nonumber \\
        &=\inf_{0 \leq s \leq r \leq 1} \sof{\mcd_\text{bin}(\midmid{r}{s})- \lambda (r-s)}.
    \end{align}
    To see the last step, recall the definition of $\mcd_\text{bin}$. This concludes the proof for $L_\mcd$. 
    
    Now, for the optimal convex lower bound recall its definition
    \begin{equation}
    B_\mcd(T (\rho, \sigma)) \coloneqq \sup_\lambda \sof{L_\mcd(\lambda) + \lambda T (\rho, \sigma )}.
    \end{equation}
    We already found that $L_\mcd$ is attained by using two-dimensional classical states. Again, we can use $\mathcal{C}_{\rho, \sigma}$ and \eqref{eq: conservation trace distance} to see that the convex lower bound, too, is attained by two-dimensional classical states.
\end{proof}

\section{Proof of Theorem \ref{thm: smoothed bounds}}\label{Appendix: Proof of smoothed theorem}
\begin{proof}
    The linear lower bound of $\mcd^\varepsilon$ is given by
    \begin{equation}
        L_{\mcd^\varepsilon}(\lambda)=\inf_{\rho, \sigma} \sof{\mcd^\varepsilon(\midmid{\rho}{\sigma})- \lambda T(\rho, \sigma)}.
    \end{equation}
    Following the arguments in the proof of Theorem \ref{thm: main theorem}, we apply the classicalization channel $\mathcal{C}_{\rho, \sigma}$ to $\rho$ and $\sigma$. Using that this channel conserves the trace distance and that the $\varepsilon$-smoothed divergence satisfies the data-processing inequality, we find
    \begin{alignat}{3}
        L_{\mcd^\varepsilon}(\lambda)&= \inf_{\mathcal{C}_{\rho, \sigma}(\rho), \mathcal{C}_{\rho, \sigma}(\sigma)} \sof{&&\mcd^\varepsilon(\midmid{\mathcal{C}_{\rho, \sigma}(\rho)}{\mathcal{C}_{\rho,\sigma}(\sigma)})\nonumber\\
            &&&- \lambda T(\mathcal{C}_{\rho, \sigma}(\rho), \mathcal{C}_{\rho, \sigma}(\sigma))}  \\[5px]
            &= \inf_{\mathcal{C}_{\rho, \sigma}(\rho), \mathcal{C}_{\rho, \sigma}(\sigma)} \sof{&&\inf_{\rho' \in B^\varepsilon(\mathcal{C}_{\rho, \sigma}(\rho))} \sof{\mcd(\midmid{\rho'}{\mathcal{C}_{\rho, \sigma}(\sigma)})} \nonumber\\
            &&&- \lambda T(\mathcal{C}_{\rho, \sigma}(\rho), \mathcal{C}_{\rho, \sigma}(\sigma))}. \label{eq: smoothed1}
    \end{alignat}
    Let us look at the two summands in the outer infimum individually and denote them by $S_1$ and $S_2$, i.\ e.\ 
    \begin{align}
        S_1 &\coloneqq \inf_{\rho' \in B^\varepsilon(\mathcal{C}_{\rho, \sigma}(\rho))} \sof{\mcd(\midmid{\rho'}{\mathcal{C}_{\rho, \sigma}(\sigma)})},\\
        S_2 &\coloneqq - \lambda T(\mathcal{C}_{\rho, \sigma}(\rho), \mathcal{C}_{\rho, \sigma}(\sigma)).
    \end{align}
    
    For the second summand $S_2$, recall that
    \begin{equation}
        \mathcal{C}_{\rho, \sigma} (\rho) = \begin{pmatrix}
            r & 0 \\
            0 & 1-r
        \end{pmatrix}, \quad
        \mathcal{C}_{\rho, \sigma} (\sigma) = \begin{pmatrix}
            s & 0 \\
            0 & 1-s
        \end{pmatrix}.
    \end{equation}
    where $0 \leq s \leq r \leq 1$. Hence, the trace distance of these two states is
    \begin{equation}
        T(\mathcal{C}_{\rho, \sigma}(\rho), \mathcal{C}_{\rho, \sigma}(\sigma)) = r - s
    \end{equation}
    and therefore
    \begin{equation}
        S_2 = - \lambda (r-s).
    \end{equation}
    
    Now, let us take a closer look at the first summand $S_1$. There the infimum is taken over the $\varepsilon$-ball around $\mathcal{C}_{\rho, \sigma}(\rho)$. This can be visualized on the Bloch ball: The qubit $\mathcal{C}_{\rho, \sigma}(\rho)$ corresponds to the Bloch vector 
    \begin{equation}
        \vec{r} = \begin{pmatrix}
            0 \\
            0 \\
            2r-1
        \end{pmatrix},
    \end{equation}
    a vector along the $z$-axis which points down for $r \in [0,\frac{1}{2})$ and up for $r \in (\frac{1}{2},1]$. The $\varepsilon$-ball around $\mathcal{C}_{\rho, \sigma}(\rho)$ corresponds to a $2\varepsilon$ Euclidean ball around this vector:
    \begin{align}
        B^\varepsilon(\mathcal{C}_{\rho ,\sigma}(\rho)) &= \sof{\rho' \in \mathcal{S}(\mathbb{C}^2) \mid T(\rho, \rho') \leq \varepsilon} \nonumber \\
        &= \bigg\{\rho' = \tfrac{1}{2}(\mathds{1} + \pvec{r}'\cdot \vec{\sigma}_\text{Pauli}) \biggm| \begin{alignedat}[]{1}
            & \norm{\vec{r} - \pvec{r}'} \leq 2\varepsilon,\\
            & \norm{\pvec{r}'} \leq 1
        \end{alignedat} \bigg\},
    \end{align}
    using the fact that the trace distance of two qubits is equal to half the Euclidean distance of their Bloch vectors. 

    Can we, without loss of generality, assume that the $\rho'$ in the infimum are diagonal? To answer this question consider the \emph{z-pinching channel} $\mathcal{C}_z:\mathcal{S}(\mathbb{C}^2) \to \mathcal{S}(\mathbb{C}^2)$ given by the Kraus operators
    \begin{equation}
        K_1 = \begin{pmatrix}
            1 & 0 \\
            0 & 0
        \end{pmatrix}, \qquad
        K_2 = \begin{pmatrix}
            0  & 0 \\
            0 & 1
        \end{pmatrix}.
    \end{equation}
    This channel pinches any qubit along the $x$ and $y$ directions, projecting its Bloch vector onto the $z$-axis of the Bloch ball:
    \begin{align}
        \mathcal{C}_z(\rho') &= \mathcal{C}_z\Bigl(\frac{1}{2}\begin{pmatrix}
            1 + z' & x' - \ii y' \\
            x' + \ii y' & 1 - z'
        \end{pmatrix} \Bigr) \nonumber \\
        &= \frac{1}{2}\begin{pmatrix}
            1 + z' & 0\\
            0 & 1 - z'
        \end{pmatrix}.
    \end{align}
    Applied to the $\varepsilon$-smoothing ball $B^\varepsilon(\mathcal{C}_{\rho, \sigma}(\rho))$, the $z$-pinching channel shrinks the ball to the $z'$-interval,
    \begin{equation}
        \mathcal{C}_z(B^\varepsilon(\mathcal{C}_{\rho ,\sigma}(\rho))) = [2r - 1 - 2\varepsilon, 2r - 1 + 2\varepsilon].
    \end{equation} 

    Noting that the $z$-pinching channel leaves $\mathcal{C}_{\rho ,\sigma}(\sigma)$ invariant and that
    \begin{equation}
        \mathcal{C}_z(B^\varepsilon(\mathcal{C}_{\rho ,\sigma}(\rho))) \subset B^\varepsilon(\mathcal{C}_{\rho, \sigma}(\rho)),
    \end{equation}
    we can use the data-processing inequality to find
    \begin{align}
        S_1 &= \inf_{\rho' \in B^\varepsilon(\mathcal{C}_{\rho ,\sigma}(\rho))} \sof{\mcd(\midmid{\mathcal{C}_z(\rho')}{\mathcal{C}_z(\mathcal{C}_{\rho ,\sigma}(\sigma)}))} \nonumber \\
        &= \inf_{\rho' \in \mathcal{C}_z(B^\varepsilon(\mathcal{C}_{\rho ,\sigma}(\rho)))} \sof{\mcd(\midmid{\rho'}{\mathcal{C}_{\rho ,\sigma}(\sigma)})},
    \end{align}
    meaning that we can indeed assume that $\rho'$ is diagonal and write it as
    \begin{equation}
        \rho' = \begin{pmatrix}
            r' & 0 \\
            0 &1-r'
        \end{pmatrix}
    \end{equation}
    with $r' \in [r - \varepsilon, r + \varepsilon]\cap[0,1]$. So far we have
    \begin{equation}
        S_1 = \inf_{r'} \sof*{\mcd\of*{\begin{pmatrix}
            r' & 0 \\
            0 &1-r'
        \end{pmatrix}\bigg\Vert\begin{pmatrix}
            s & 0 \\
            0 &1-s
        \end{pmatrix}}}.
    \end{equation}

    There are two cases to consider for the infimum: $\abs{r - s } \leq \varepsilon$ and $\abs{r - s} > \varepsilon$. In the first case the infimum is clearly attained at $r=s$ or rather $\rho' = \mathcal{C}_{\rho, \sigma}(\sigma)$ because $\mcd$ vanishes whenever the two states coincide and since it is non-negative, zero is the smallest value it can reach. In the second case we note that $\mcd$ is convex in the first argument\footnote{Most divergences are jointly convex in both arguments. Joint convexity implies convexity in each argument \cite{Boyd_2004}.} and together with non-negativity we find that $\mcd$ is monotonically non-decreasing in the first argument. Therefore, the infimum is attained at the smallest possible value for $r'$, which is $r' = r - \varepsilon$. Together we have found 
    \begin{equation}
        S_1 = \begin{cases}
            0 &\text{for } \abs{r-s} \leq \varepsilon, \\
            \mcd_\text{bin}(\midmid{r-\varepsilon}{s}) &\text{for }\abs{r-s} > \varepsilon.
        \end{cases}
    \end{equation}

    With these two terms at hand, we can now evaluate the outer infimum in \eqref{eq: smoothed1} to determine the linear lower bound. In the case $\abs{r - s } \leq \varepsilon$ the linear bound vanishes and so we focus on the case $\abs{r - s} > \varepsilon$. Let the transition point between the two cases be $\lambda_\varepsilon$, which is found by evaluating the derivative of the convex bound at $\varepsilon$. For $\lambda > \lambda_\varepsilon$ we have so far obtained
    \begin{equation}
        L_{\mcd^\varepsilon}(\lambda)= \inf_{\substack{0 \leq s \leq r \leq 1\\ \abs{r - s} > \varepsilon}} \sof{\mcd_\text{bin}(\midmid{r-\varepsilon}{s})- \lambda (r - s)}. 
    \end{equation}
    Substituting $r' = r - \varepsilon$ we find
    \begin{equation}
        L_{\mcd^\varepsilon}(\lambda)= \inf_{0 \leq s \leq r' \leq 1} \sof{\mcd_\text{bin}(\midmid{r'}{s})- \lambda (r' - s)} - \lambda \varepsilon,
    \end{equation}
    which is just the linear bound for the unsmoothed divergence shifted by $- \lambda \varepsilon$, namely
    \begin{equation}
        L_{\mcd^\varepsilon}(\lambda)= L_\mcd(\lambda) - \lambda \varepsilon. 
    \end{equation}
    Together, the linear lower bound is
    \begin{equation}
        L_{\mcd^\varepsilon}(\lambda)= \begin{cases}
            0 & \text{for } 0 \leq \lambda \leq \lambda_\varepsilon, \\
            L_\mcd(\lambda) - \lambda \varepsilon &\text{for } \lambda_\varepsilon < \lambda \leq \lambda_\text{max},
        \end{cases}
    \end{equation}
    where $\lambda_\varepsilon = \left. \frac{\dd B_\mcd}{\dd T}\right\vert_{T=\varepsilon}$ and $\lambda_\text{max} = \left.\frac{\dd B_\mcd}{\dd T}\right\vert_{T=1-\varepsilon}$.
    
    Finding the convex lower bound is simpler than finding the linear bound because we can build on the results obtained so far. The convex lower bound is found through evaluating
    \begin{equation}
        B_{\mcd^\varepsilon} = \sup_\lambda \sof{L_{\mcd^\varepsilon} + \lambda T}.
    \end{equation}
    In the case $\abs{r - s } = T \leq \varepsilon$, the convex bound is zero because in that case, the minimal smoothed divergence is attained where the two states are equal. In the case $\abs{r - s} = T > \varepsilon$, we find
    \begin{align}
        B_{\mcd^\varepsilon} &= \sup_\lambda \sof{L_{\mcd^\varepsilon} + \lambda T} \nonumber \\
        &= \sup_\lambda \sof{L_{\mcd} - \lambda \varepsilon + \lambda T} \nonumber \\
        &= \sup_\lambda \sof{L_{\mcd} + \lambda( T -\varepsilon )} \nonumber \\
        &= B_\mcd(T-\varepsilon).
    \end{align}
    Thus, the convex bound of the $\varepsilon$-smoothed divergence is just the convex bound of the unsmoothed divergence shifted by $\varepsilon$ and set to zero where the convex bound of the unsmoothed divergence is negative
    \begin{equation}
        B_{\mcd^\varepsilon}(T)= \begin{cases}
            0 & \text{for } 0 \leq T \leq \varepsilon,\\
            B_\mcd(T - \varepsilon) &\text{for }\varepsilon < T \leq 1.
        \end{cases}
    \end{equation}
\end{proof}

\section{Calculation of analytical bounds}
We discussed the calculation of both the linear and the convex bound for the Rényi divergence for $\alpha \in [1,\infty]$ in length in Section \ref{sec: Example Renyi divergence}. The other analytical bounds in Table \ref{tab: results} are obtained in a similar fashion. Here we give the important steps to obtain these bounds.

\subsection{Hellinger Divergence}
The Hellinger divergence of order $\alpha$ is the $f$-divergence for $f_\alpha = \frac{x^\alpha - 1}{\alpha - 1}$. Its binary version is 
\begin{equation}
    D_{H_\alpha, \text{bin}} (\midmid{r}{s}) = \tfrac{1}{\alpha - 1 }\of*{r^\alpha s^{1-\alpha} + (1-r)^\alpha (1-s)^{1-\alpha}-1}.
\end{equation}
Note that the binary versions of the Hellinger divergence of order $\alpha$ and the Rényi divergence of order $\alpha$ are related \cite{Beigi_2025} by
\begin{equation}
    D_{\alpha, \text{bin}} (\midmid{r}{s}) =\tfrac{1}{\alpha - 1}\log\of*{1 + (\alpha - 1 )D_{H_\alpha, \text{bin}} (\midmid{r}{s})}.
\end{equation}

To find the linear bound $L_{D_{H_\alpha}}$, we need to solve the optimization problem
\begin{align}
       &L_{D_{H_\alpha}}(\lambda)=\inf_{r,s} \sof{D_{H_\alpha, \text{bin}}(\midmid{r}{s})-\lambda (r-s) }\quad \\ &\operatorname{s.th. } 0\leq s \leq r \leq 1.\nonumber
\end{align}
We set
\begin{equation}
    \xi_\lambda (r,s)\coloneqq D_{H_\alpha, \text{bin}}(\midmid{r}{s})- \lambda \abs{r-s}.
\end{equation}
As before in the calculation for the bounds of the collision entropy in Section \ref{sec: Example Renyi divergence}, there are two cases to be  considered. In the case where the minimum is not at the boundary, we suspect that the bounds cannot be calculated analytically and instead calculate them numerically, see Section \ref{sec: numerics}. In the case where the minimum is at the boundary, we can calculate the bounds analytically. To do so, we solve
\begin{equation}
    \frac{\partial \xi_{\lambda} (r=1,s)}{\partial s} \overset{!}{=} 0
\end{equation}
to find the position of the minimum at
\begin{equation}
    (r^*, s^*) = \of*{1, \lambda^{-\frac{1}{\alpha}}}.
\end{equation}
Hence, our linear bound is
\begin{align}
    L_{D_{H_\alpha}}(\lambda) &= \xi_{\lambda} (r^*, s^*) \nonumber \\
    &= \frac{1}{1-\alpha}\of*{1+\lambda(\alpha - 1 - \alpha \lambda^{-\frac{1}{\alpha}})}.
\end{align}
This bound is valid for $\lambda \geq \lambda_\text{crit}$, where $\lambda_\text{crit} = \of*{\frac{\alpha}{\alpha-1}}^\alpha$ is found by solving
\begin{equation}
    \frac{\partial \xi_{\alpha, \lambda}(r,s)}{\partial r}\rvert_{s=s^*, r=1} \overset{!}{=} 0.
\end{equation}

The convex bound $B_{D_{H_\alpha}}$ is then obtained by solving 
\begin{equation}
    B_{D_{H_\alpha}}(T) = \sup_\lambda \sof{L_{D_{H_\alpha}}(\lambda) + \lambda T},
\end{equation}
which is done through
\begin{equation}
    \frac{\partial}{\partial \lambda}(L_{D_{H_\alpha}}(\lambda) + \lambda T) \overset{!}{=} 0,
\end{equation}
where we find
\begin{equation}
    \lambda^* = \of*{\frac{1}{1-T}}^\alpha
\end{equation}
and hence
\begin{align}
    B_{D_{H_\alpha}}(T) &= L_{D_{H_\alpha}}(\lambda^*) + \lambda^* T \nonumber \\
        &= \frac{1-(1-T)^{1-\alpha}}{1-\alpha}
\end{align}
for $T \in [\frac{1}{\alpha}, 1]$.

Note that for the case $\alpha = 1$, one has to consider the limit $\alpha \to 1$.

\subsection{Neyman-\texorpdfstring{$\chi^2$}{chi squared}}
The Neyman-$\chi^2$ is the $f$-divergence for $f(x) = (x-1)^2$. Its binary version is
\begin{equation}
    \chi_{N, \text{bin}}^2 (\midmid{r}{s}) = \frac{(r-s)^2}{s(1-s)}.
\end{equation}
To find the linear bound $L_{\chi_N^2}$, we need to solve the optimization problem
\begin{align}
       &L_{\chi_N^2}(\lambda)=\inf_{r,s} \sof{\chi_{N, \text{bin}}^2(\midmid{r}{s})-\lambda (r-s) }\quad \\ &\operatorname{s.th. } 0\leq s\leq r \leq 1.\nonumber
\end{align}
Again, like in the case of the Hellinger divergence, we have to consider two cases. The Neyman-$\chi^2$ is equal to the Hellinger divergence of order 2 and we can therefore use the bounds of the Hellinger divergence for the case where the minimum is at the boundary. These are
\begin{align}
    L_{\chi_N^2}(\lambda) &= \frac{1}{1-2}\of*{1+\lambda(2 - 1 - 2\lambda^{-\frac{1}{2}})} \nonumber \\
    &= -(\sqrt{\lambda}-1)^2
\end{align}
for $\lambda \geq 4$ and
\begin{align}
    B_{\chi_N^2}(T) &= \frac{1-(1-T)^{1-2}}{1-2} \nonumber \\
        &= \frac{T}{1-T}
\end{align}
for $T \in [\tfrac{1}{2}, 1]$.

In the case where the minimum is not at the boundary and rather lies inside the lower right triangle, we set
\begin{equation}
    \xi_\lambda (r,s)\coloneqq \chi_{N, \text{bin}}^2(\midmid{r}{s})- \lambda \abs{r-s} 
\end{equation}
and solve
\begin{equation}
    \vec{\nabla} \xi_\lambda \overset{!}{=} 0
\end{equation}
to find the position of the minimum at
\begin{equation}
    (r^*, s^*) = \of*{\frac{1}{2}\of*{1+\frac{\lambda}{4}}, \frac{1}{2}}.
\end{equation}
Hence, our linear bound is
\begin{align}
    L_{\chi_N^2}(\lambda) &= \xi_{\lambda} (r^*, s^*) \nonumber \\
    &=-\frac{\lambda^2}{16}.
\end{align}
This bound is valid for $\lambda \in [0,4]$.

The convex bound $B_{\chi_N^2}$ is found by solving 
\begin{equation}
    \frac{\partial}{\partial \lambda}(L_{\chi_N^2}(\lambda) + \lambda T) \overset{!}{=} 0,
\end{equation}
where we find
\begin{equation}
    \lambda^* = 8T
\end{equation}
and thus
\begin{align}
        B_{\chi_N^2}(T) &= L_{\chi_N^2}(\lambda^*) + \lambda^* T \nonumber \\
        &= 4 T^2,
\end{align}
for $T \in [0, \frac{1}{2}]$.

\subsection{Pearson-\texorpdfstring{$\chi^2$}{chi squared}}
The Pearson-$\chi^2$ is also an $f$-divergence, namely for $f(x) = \frac{(x-1)^2}{x}$. Its binary version is given by
\begin{equation}
    \chi_{P, \text{bin}}^2 (\midmid{r}{s}) = \frac{(r-s)^2}{r(1-r)}.
\end{equation}
Note that $\chi_{P, \text{bin}}^2$ is the same as $\chi_{N, \text{bin}}^2$ but with $r$ and $s$ swapped. Since $L_{\chi_{P}^2}$ is obtained by solving an optimization  problem over both $r$ and $s$, the result for $L_{\chi_{P}^2}$ is the same as for the Neyman-$\chi^2$. Consequently, the convex bound is also the same for both divergences.

\subsection{Fidelity Divergence}\label{Appendix: alpha 1/2}
Recall the binary version of the Rényi divergence
\begin{equation}
    D_{\alpha, \text{bin}} = \frac{1}{\alpha-1}\log(r^\alpha s^{1-\alpha} + (1-r)^\alpha (1-s)^{1-\alpha}).
\end{equation}
For the case $\alpha = \frac{1}{2}$, this simplifies to
\begin{equation}
    D_{\frac{1}{2}, \text{bin}} = -2 \log\of*{\sqrt{rs} + \sqrt{(1-r)(1-s)}}.
\end{equation}

The first step is to find the linear lower bound using Theorem \ref{thm: main theorem}, i.e.\ ,
\begin{align}
       &L_{D_\frac{1}{2}}(\lambda)=\inf_{r,s} \sof{D_{\frac{1}{2},\text{bin}}(\midmid{r}{s})- \lambda (r-s) }\quad \\ &\text{s.th. } 0\leq r,s\leq 1.\nonumber
\end{align}

\begin{figure}
  \centering
  \includegraphics{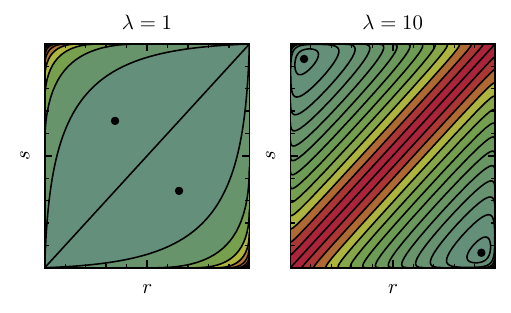}
  \caption{\label{fig: xi_fidelita_divergence}$\xi_{\frac{1}{2}, \lambda}$ on the unit square for different values of $\lambda$. The minima (black dots) lie on the line $s=1-r$.}
\end{figure}

Set
\begin{align}
    \xi_\lambda (r,s) &\coloneqq D_{\frac{1}{2}, \text{bin}}(\midmid{r}{s})- \lambda \abs{r-s} \nonumber \\
    &\phantomcolon = -2 \log\of*{\sqrt{rs} + \sqrt{(1-r)(1-s)}}- \lambda \abs{r-s}.
\end{align}
Note that in addition to the symmetry $(r,s)\mapsto(1-r,1-s)$ that we have by construction, we also have the symmetry $(r,s)\mapsto (1-s,1-r)$. Figure \ref{fig: xi_fidelita_divergence} shows that the minimum lies on the line $s = 1-r$. If we restrict $\xi_\lambda$ to this line, it simplifies to
\begin{equation}
    \xi_\lambda (r,1-r) = \log\of*{\frac{1}{4r(1-r)}} - \lambda \abs{2r-1}.
\end{equation}
We find the minimum by solving
\begin{equation}
    \frac{\partial\xi_\lambda (r,1-r)}{\partial r} \overset{!}{=}0
\end{equation}
to be at
\begin{align}
    r^* &= \tfrac{-2 + \lambda \ln(4) \pm \sqrt{4 + \lambda^2 \ln(4)^2}}{2 \lambda \ln(4)}, \nonumber \\
    s^* &= 1-r^*,
\end{align}
where only the solution with the positive square root is relevant for us because the other solution results in an imaginary bound $B_{D_\frac{1}{2}}$. The linear bound then is
\begin{alignat}{1}
    L_{D_\frac{1}{2}}(\lambda) = \frac{1}{\ln(2)} \bigg(&1-\sqrt{1+\lambda^2\ln(2)^2}\nonumber\\
    &+\ln\of*{\tfrac{1+\sqrt{1+\lambda^2\ln(2)^2}}{2}} \bigg).
\end{alignat}

The second step is to obtain the convex lower bound given through
\begin{equation}
    B_{D_\frac{1}{2}}(T) = \sup_\lambda \sof{L_{D_{\frac{1}{2}}}(\lambda) + \lambda T}.
\end{equation}
Solving
\begin{equation}
    \frac{\partial}{\partial \lambda}(L_{D_\frac{1}{2}}(\lambda) + \lambda T) \overset{!}{=} 0,
\end{equation}
we find
\begin{equation}
    \lambda^* = \frac{2T}{\ln(2)(1-T^2)}
\end{equation}
and hence
\begin{equation}
    B_{D_\frac{1}{2}}(T) = \log\of*{\frac{1}{1-T^2}},
\end{equation}
for all $T \in [0,1]$.

\subsection{Umegaki Divergence}\label{Appendix: relative entropy}
The binary version of the Umegaki divergence is 
\begin{equation}
    D_\text{bin}(\midmid{r}{s})= \log\of*{\of*{\frac{r}{s}}^r \Bigl(\frac{1-r}{1-s}\Bigr)^{1-r}}.
\end{equation}
For the Umegaki divergence, the minimum 
\begin{align}
       &L_{D}(\lambda)=\inf_{r,s} \sof{D_{\text{bin}}(\midmid{r}{s})- \lambda (r-s) }\quad \\ &\operatorname{s.th. } 0\leq r,s\leq 1\nonumber
\end{align}
never lies at the boundary.  We find it to be at
\begin{align}
    r^* &= \tfrac{2^\lambda(2^\lambda - 1 - \lambda \ln(2))} {(2^\lambda-1)^2}, \\
    s^* &= \tfrac{1}{1-2^\lambda} + \tfrac{1}{\lambda \ln(2)},
\end{align}
giving the linear bound
\begin{equation}
    L_D (\lambda) = \log\of*{\of*{\frac{r^*}{s^*}}^{r^*} \of*{\frac{1-r^*}{1-s^*}}^{1-r^*}}.
\end{equation}

Calculating the convex linear bound analytically is not possible to our knowledge. However, following the argumentation in \cite{Fedotov}, we can obtain a parametrized convex bound as follows: The convex conjugate of $D_\text{bin}$ is given as
\begin{equation}
    D_\text{bin}^*(x, y) \coloneqq \sup_{r,s} \sof*{\begin{pmatrix}x \\ y\end{pmatrix} \cdot \begin{pmatrix}r \\ s\end{pmatrix} - D_\text{bin}(\midmid{r}{s})}.
\end{equation}
Analogously to our calculation of $L_D$, we find the supremum at
\begin{align}
    r^* &=  \frac{2^x(2^x - 1 + y \ln(2))}{(2^x - 1)^2}, \\
    s^* &= \frac{1}{1 - 2^x} - \frac{1}{y \ln(2)},
\end{align}
for $\frac{1-2^x}{\ln(2)} < y < - \frac{1-2^{-x}}{\ln(2)}$. Using 
\begin{equation}
    \begin{pmatrix} x \\ y \end{pmatrix} =  \begin{pmatrix} t \\ -t \end{pmatrix},
\end{equation}
we find 
\begin{equation}
    \begin{pmatrix} x \\ y \end{pmatrix} \cdot  \begin{pmatrix} r \\ s \end{pmatrix} = t \abs{r-s} = t T
\end{equation}
and hence
\begin{align}
    D_\text{bin}^*(t, -t) &= \sup_{T} \of*{tT - D(V)}, \\
    T &= \frac{xr+ys}{t}.
\end{align}
Plugging in $(r^*,s^*)$ we have
\begin{align}
    T(t) &= \left. \frac{x r^* + y s^*}{t}\right\vert_{x=t, y=-t} \nonumber \\
    &= \frac{(2^t-1-\ln(2^t))(1-2^t+2^t\ln(2^t))}{(2^t-1)^2\ln(2^t)}
\end{align}
and
\begin{align}
    D(t) &= \left. \log\of*{\of*{\frac{r^*}{s^*}}^{r^*} \of*{\frac{1-r^*}{1-s^*}}^{1-r^*}}\right\vert_{x=t, y=-t} \nonumber \\
    &= \frac{1-2^t+2^t\ln(2^t)}{(2^t-1)^2}\log\of*{\tfrac{t}{2^t-1}} \nonumber\\ &\quad-2^t\frac{1-2^t+\ln(2^t)}{(2^t-1)^2} \log\of*{\tfrac{2^t t}{2^t-1}}\nonumber \\
    &\quad + \log(\ln(2)),
\end{align}
where $t \geq 0$.

\end{document}